\documentclass[11p,reqno]{amsart}

\usepackage[T1]{fontenc}
\usepackage{graphicx}
\usepackage{amssymb,amsthm,amsmath,mathrsfs,bm,braket,marginnote}
\usepackage{enumerate}
\usepackage{appendix}
\usepackage{bbm}
\usepackage{eufrak}
\usepackage{mathrsfs}
\usepackage[colorlinks=true, pdfstartview=FitV, linkcolor=blue, citecolor=blue, urlcolor=blue]{hyperref}

\usepackage{pgf}
\usepackage{pgfplots}
\usepackage{tikz}
\usetikzlibrary{arrows,calc}
\usepackage{verbatim}
\usetikzlibrary{decorations.pathreplacing,decorations.pathmorphing}

\numberwithin{equation}{section}
\newtheorem{theorem}{Theorem}[section]
\newtheorem{lemma}[theorem]{Lemma}
\newtheorem{proposition}[theorem]{Proposition}

\newtheorem{remark}[theorem]{Remark}
\newtheorem{corollary}[theorem]{Corollary}

\newcommand{\p}{\partial}
\newcommand{\Pf}{\text{Pf}}

\newcommand{\tp}{\tilde{P}}
\newcommand{\ts}{\tilde{S}}
\newcommand{\xitt}{\xi(t'-t,z^{-1})}
\newcommand{\xit}{\xi(t-t',z)}
\newcommand{\di}{\text{diag}}
\newcommand{\htau}{\hat{\tau}}

\newcommand{\hr}{\hat{\rho}}
\newcommand{\hs}{\hat{\sigma}}

 \reversemarginpar
\begin{document}

\title[Rank shift conditions and reductions of 2d-Toda theory]{Rank shift conditions and reductions of 2d-Toda theory}
\author{Shi-Hao Li}
\address{ School of Mathematical and Statistics, ARC Centre of Excellence for Mathematical and Statistical Frontiers, The University of Melbourne, Victoria 3010, Australia}
\email{shihao.li@unimelb.edu.au}
\author{Guo-Fu Yu}
\address{School of Mathematical Sciences, Shanghai Jiaotong University, People's Republic of China.}
\email{gfyu@sjtu.edu.cn}

\date{}

\dedicatory{}

\begin{abstract}
This paper focuses on different reductions of 2-dimensional (2d-)Toda hierarchy. Symmetric and skew symmetric moment matrices are firstly considered, resulting in the differential relations between symmetric/skew symmetric tau functions and 2d-Toda's tau functions respectively. Furthermore, motivated by the Cauchy two-matrix model and Bures ensemble from random matrix theory, we study the rank one shift condition in symmetric case and rank two shift condition in skew symmetric case, from which the C-Toda hierarchy and B-Toda hierarchy are found, together with their special Lax matrices and integrable structures.
\end{abstract}

\maketitle
\section{Introduction}
The studies in random matrix theory and classic integrable systems promote the development of both communities, and the orthogonal polynomials play an important role to connect these two totally different subjects. For example, in \cite[\S 2]{deift00}, it is shown that these polynomials can lead to an iso-spectral flow of the Toda lattice, while the normalisation factors of the orthogonal polynomials can be expressed in terms of the partition function of the unitary invariant Hermitian matrix model. Later on, several novel random matrix models have been found in the course of doing analysis of the classical integrable systems. An example is the appearance of the Cauchy two-matrix model. This model was proposed in the studies of the peakon solutions of the Degasperis-Procesi equation and its related Hermite-Pad\'e approximation problem \cite{lundmark2005}. 

To some extent, the relation between random matrix and integrable system can be described in terms of the partition function and tau functions. It has been realised that the partition function of the coupled two-matrix model could be regarded as the tau function of 2d-Toda hierarchy \cite{adler19992}, and the partition function of symplectic ensemble and orthogonal ensemble could be thought as the tau functions of Pfaff lattice and some differential-difference equations \cite{adler00,adler2002,hu06}. It was also found that the partition function of Bures ensemble could be viewed as the tau function of BKP hierarchy \cite{chang182,hu17,orlov16} and that of the Cauchy two-matrix model is related to the CKP hierarchy \cite{li2019}. The average characteristic polynomials (some of them are not orthogonal!) in matrix models are used to find the correspondence \cite{chang182,li2019}. Though effective, only few members of these hierarchies could be explicitly computed. For this reason, another method called the ``moment matrix approach'' is employed to conquer the difficulty in this paper, to reveal more information about these integrable hierarchies behind the Cauchy two-matrix model as well as the Bures ensemble.

The moment matrix approach is based on the group factorization theory and iso-spectral analysis, developed by Adler and van Moerbeke. An overall review could be seen in \cite{adler18}. The basic idea is that from the moment matrix, one can obtain the dressing operators, wave functions and tau functions by the matrix decomposition. To the best of our knowledge, the most general case was considered in \cite{adler19992}, motivated by the coupled two matrix model, 
showing that the tau functions of 2d-Toda hierarchy could be generated by a very general bi-moment matrix $m_\infty$
with evolutions $\p_{t_n}m_\infty=\Lambda^n m_\infty$ and $\p_{s_n}m_\infty=-m_\infty \Lambda^{\top n}$.
Later on, some reduction theory about the 2d-Toda hierarchy have been considered, motivated by different random matrix models. In the literatures, there have been some examples and great success has been achieved. For instance,
\begin{itemize}
\item As observed in \cite{adler97}, the constraint $\Lambda^i m_\infty=m_\infty \Lambda^{\top i}$ leads the 2d-Toda hierarchy to the 1d-Toda hierarchy, in which the time flows $v_n=t_n+s_n$ keep invariant and $u_n=t_n-s_n$ would make contributions. In fact, this constraint on the moment matrix indicates that the matrix elements satisfy the relation $m_{k+i,l}=m_{k,l+i}$, amounting to a Hankel matrix. It is well known that the 1d-Toda's tau functions could be realised by the time-dependent partition function of the Hermitian matrix model with unitary invariance \cite{gerasimov90}, which could also be expressed in terms of the Hankel determinants;

\item The case considered in \cite{adler1999, adler2002} is the skew symmetric reduction on the moment matrix such that $m_\infty^\top=-m_\infty$. Obviously, the tau functions defined by the skew symmetric moment matrix should be written as Pfaffians rather than determinants due to the skew symmetry. Applying the commuting flows $\p_{t_n}m_\infty=\Lambda^n m_\infty+m_\infty \Lambda^{\top n}$ to the moment matrix leads to the Pfaff lattice via the skew-Borel decomposition and the integrability follows from the Adler-Kostant-Symes theorem \cite{adler1999,adler18}. The Dyson's symplectic and orthogonal ensembles play important roles in this case. Moreover, the interrelation between Dyson's $\beta$-ensemble when $\beta=1,2,4$ implies the relationship between the Toda lattice and Pfaff lattice via the dressing operator and spectral problem, see \cite{adler00,adler02};
\item The Toeplitz matrix/determinant appears in the studies of the unitary invariant ensemble on the unit circle and random permutations, which shows the reduction on the moment matrix $\Lambda m_\infty \Lambda^{\top}=m_\infty$ (or equally $m_{i,j}:=m_{i-j}$). In fact, the Toeplitz moment matrix is closely related to the integrable systems. If the evolutions $\p_{t_n}m_i=m_{i+n}$ and $\p_{s_n}m_i=-m_{i-n}$ are permitted, there are some knwon examples including the $2+1$ dimensional generalisation of Ablowitz-Ladik lattice and the $2+1$ dimensional AKNS hierarchy \cite{adler03}, allowing some further interplays with the Calabi-Yau threefolds \cite{brini17}. Moreover, if the time evolutions $\p_{t_n}m_{i}=m_{i+n}+m_{i-n}$ are permitted, one can find the corresponding Schur flow \cite{faybusovich99,forrester06,mukaihira02}.  
\end{itemize}
The considerations of subsequent developments within the random matrix theory suggest further interplays. In this paper, we focus on two different random matrix models and corresponding integrable hierarchies. The first one is the Cauchy two-matrix model defined on the configuration space $\mathbb{R}_+\times\mathbb{R}_+$, as specified by an eigenvalue probability density function \cite{bertola2009,forrester16}
\begin{align*}
\frac{\prod_{1\leq i<j\leq n}(x_j-x_i)^2(y_j-y_i)^2}{\prod_{i,j=1}^n (x_i+y_j)}\prod_{i,j=1}^n\omega_1(x_i)\omega_2(y_j)
\end{align*}
with some non-negative weight functions $\omega_1$ and $\omega_2$. This model could be considered as a two-field model, where the interaction of the fields is described by the Cauchy kernel.
It has been shown that in the symmetric case $\omega_1=\omega_2$, the time-dependent partition function of this model could be regarded as the tau function of the CKP hierarchy if we introduce an infinite series of time parameters $\{t_i\}_{i\in\mathbb{N}}$ into the weight function and assume that $\p_{t_i}\omega(x;t)=x^i\omega(x;t)$ \cite{li2019}. At the same time, there is an integrable lattice called the C-Toda lattice, whose tau function can be expressed in terms of the Cauchy two-matrix model. It has also been found that the quantity of the partition function could be expressed as a Gram determinant with symmetric elements. Therefore, it motivates us to consider the (semi-)discrete hierarchy if the initial moment matrix is assumed to be symmetric. It would be demonstrated that the symmetric reduction is so strong that the bilinear form of tau functions would be broken. Luckily, the moments of the Cauchy two-matrix model give us another constraint, called the rank one shift condition, finally resulting in the C-Toda hierarchy.

The other model closely related to the Cauchy-two matrix model is the Bures ensemble \cite{bures69}. The significances of this model are on the two folds in physics. One is in quantum information theory---this model is a metric induced ensemble, induced by the Bures measure \cite{sommers03}. The other is in random matrix theory referred to as the  $O(1)$ model \cite{bertola2009}, which is a two-matrix model for $n\times n$ positive definite Hermitian matrices $M$ and $A$ with distribution 
\begin{align*}
\exp\left(
-n\text{Tr}(V(M)+MA^2)
\right)dMdA,
\end{align*}
where the potential function $V$ is assumed to be non-negative. The integration over the Gaussian variable $A$ can be performed and thus yield a measure over the eigenvalues $\{x_i,\,i=1,\cdots,n\}$ of $M$ on $\mathbb{R}^n_+$, admitting the form
\begin{align*}
\prod_{1\leq i<j\leq n}\frac{(x_j-x_i)^2}{x_j+x_i}\prod_{i=1}^ne^{-nV(x_i)}.
\end{align*}
The time-dependent partition function of this model could be regarded as the tau function of BKP hierarchy \cite{hu17,orlov16} as well as several integrable lattices \cite{chang182}, including the interesting B-Toda lattice considered in \cite{hirota01}. Moreover, as tau functions could be written into Pfaffian, it inspires us to consider the skew symmetric reduction on the moment matrix and furthermore to find out the whole hierarchy of the B-Toda lattice. Although skew symmetric reduction has already been considered in \cite{adler1999,adler02,adler2002}, the known results are related to the DKP hierarchy (or Pfaff lattice hierarchy) according to the classification by the Kyoto school \cite{jimbo1983}. The connections between the Pfaff lattice hierarchy and the large BKP hierarchy proposed by Kac and van de Leur \cite{kac97} was left as an ongoing problem in \cite{adler2002}. Therefore, in this paper, we'd like to show that by introducing odd-indexed Pfaffian tau functions, the large BKP hierarchy is the same with Pfaff lattice hierarchy in the sense of wave functions.  Furthermore, the moments of the Bures ensemble would indicate a special rank two shift condition and we would give a novel sub-hierarchy of the Pfaff lattice with features of the BKP hierarchy. 

The paper is organized as follows. In Section \ref{sec:review}, we give a brief review of the moment matrix decomposition and group factorisation which is used for the upcoming sections. The symmetric reduction is considered in Section \ref{sec:symmetric}. We firstly consider the symmetric moment matrix and see what would happen if the time evolutions satisfy the commuting time flows $\p_{t_n}m_\infty=\Lambda^nm_\infty+m_\infty\Lambda^{\top n}$. It will be shown that the bilinear equation for the wave functions are quite simple, but it is not the case for tau functions---the symmetric reduction breaks down the bilinear form. Fortunately, the rank one shift condition would lead us to a local spectral problem, and further to a C-Toda hierarchy with auxiliary variables.
In Section \ref{sec:skewsymmetric}, we consider the skew symmetric moment matrix and Pfaffian tau functions. The most general skew symmetric reduction and the corresponding Pfaff lattice are summarised in \ref{4.1}. In \ref{4.2}, we would like to demonstrate an equivalent skew symmetric decomposition, resulting in the partial skew orthogonal polynomials and odd-indexed wave functions. With the introduction of odd-indexed Pfaffian tau functions, it is shown that the Pfaff lattice is exactly the large BKP hierarchy. However, the tau functions of the Pfaff lattice hierarchy lack the features of $B$-type equation---the Gram-type Pfaffian elements. Therefore, in the following part, we take the rank two shift condition into account and demonstrate how to derive the B-Toda hierarchy from the 2d-Toda theory. Apart from being embedded into 2d-Toda theory, the rank two shift condition implies a particular Fay identity, indicating the recurrence relation for the partial skew orthogonal polynomials, and thus give a novel sub-hierarchy of Pfaff lattice hierarchy.

\section{Review of 2d-Toda theory}\label{sec:review}
For self-consistency, it is necessary to give a brief review to 2d-Toda theory since the article is mainly about the reductions of 2d-Toda theory. This section includes the following three parts: moment matrix decomposition, time evolutions and wave functions and tau functions. Please refer to \cite{adler19992,adler2002,takasaki18} for more details.

\subsection{Moment Matrix Decomposition}
Moment matrix is the linking node of different mathematical objects such as integrable systems, orthogonal polynomials and random matrix theory. We call a semi-infinite matrix $m_\infty=(m_{i,j})_{i,j\in\mathbb{N}}$ a moment matrix if all of its principal minors are nonzero. Attributed to this fact, the moment matrix admits the LU decomposition (Gauu-Borel decomposition) $m_\infty=S_1^{-1}S_2$, where\footnote{Please note we mainly consider the semi-infinite case throughout the paper, given by the algebra $\mathcal{A}^{\frac{\infty}{2}}=\{(s_{i,j})_{i,j\in\mathbb{N}},s_{i,j}\in\mathbb{C}\}$. Usually we denote the shift operator $\Lambda=(\delta_{i+1.j})_{i,j\in\mathbb{N}}$ and $\Lambda^{-1}=\Lambda^\top=(\delta_{i,j+1})_{i,j\in\mathbb{N}}$. The shift operator and its inverse are connected by $\Lambda\Lambda^{-1}=I$ and $\Lambda^{-1}\Lambda=I-E$ with $E_{i,j}=\delta_{i,0}\delta_{j,0}$. In this case, $S_1$ and $S_2$ can be expressed in a formal way as
$S_1=\Lambda^0+\Lambda^{-1}a_{-1}+\Lambda^{-2}a_{-2}+\cdots$ and $S_2=\Lambda^0a_0+\Lambda^1 a_1+\Lambda^2 a_2+\cdots$ with $a_i,\, i\in\mathbb{Z}$ vectors.}
\begin{subequations}
\begin{align}
S_1\in& \mathfrak{G}_-=\{\text{lower triangular matrices with diagonals $1$}\},\label{g-}\\
S_2\in \mathfrak{G}_+&=\{\text{upper triangular matrices with nonzero diagonals}\},\label{g+}
\end{align}
\end{subequations}
and the corresponding Lie algebras $\mathfrak{g}_+$ and $\mathfrak{g}_-$. From the basic computation of matrix decomposition \cite{manas19}, one could compute
\begin{align*}
S_1=\left(\begin{array}{cccc}
1&&&\\
-\tau_{10}/\tau_1&1&&\\
\tau_{20}/\tau_2&-\tau_{21}/\tau_2&1&\\
\vdots&\vdots&\vdots&\ddots
\end{array}
\right)
\end{align*}
with $\tau_n=\det(m_{i,j})_{i,j=0}^{n-1}$ and $\tau_{ij}$ the determinant for the $i$-th principal minor add the next row and delete its $j$-th row. For example, 
$$\tau_{10}=m_{10},\quad \tau_{20}=\det\left|\begin{array}{cc}
m_{10}&m_{11}\\
m_{20}&m_{21}
\end{array}
\right|,\quad \tau_{21}=\det\left|\begin{array}{cc}
m_{00}&m_{01}\\
m_{20}&m_{21}
\end{array}
\right|,\quad \cdots.
$$
The diagonal of $S_2$ is important and we denote it as $h$, where $$h=\left(
\frac{\tau_1}{\tau_0},\frac{\tau_2}{\tau_1},\cdots
\right),\quad \text{ with $\tau_0=1.$}$$
Furthermore, if we define Lax matrices $L_1=S_1\Lambda S_1^{-1}$ and $L_2=S_2\Lambda^\top S_2^{-1}$, then $P(x)=S_1\chi(x)$, $Q(y)=S_2^{-\top}\chi(y)$ with $\chi(x)=(1,x,x^2,\cdots)$ are the eigenfunctions of $L_1$ and $L_2$ respectively, enjoying the spectral problems
\begin{align*}
L_1P(x)=xP(x),\quad L_2^\top Q(y)=yQ(y)
\end{align*}
and orthogonal relation
\begin{align*}
\langle P(x), Q^\top(y)\rangle=I,
\end{align*}
where $I$ is the identity matrix and the inner product is defined by $\langle x^i,y^j\rangle=m_{i,j}$.

\subsection{Time Evolutions}
Let's consider a time-dependent moment matrix $m_\infty(t,s)$, which admits the time evolutions
\begin{align*}
\frac{\p m_\infty}{\p t_n}=\Lambda^n m_\infty,\quad \frac{\p m_\infty}{\p s_n}=-m_\infty \Lambda^{\top n}.
\end{align*}
Given the initial value $m_\infty(0,0)$, the unique solution of the evolution equation can be written as
\begin{align}\label{solution}
m_\infty(t,s)=e^{\sum_{i=1}^\infty t_i\Lambda^i}m_{\infty}(0,0)e^{-\sum_{i=1}^\infty s_i\Lambda^{\top i}}.
\end{align}
Obviously, if $m_\infty(t,s)$ is non-degenerate for any parameters $t,\,s\in \mathbb{R}$, then the Gauss-Borel decomposition is valid, and $\mathfrak{G}_\pm$ defined in \eqref{g-} and \eqref{g+} form the Lie subgroups of the general linear group with corresponding Lie algebras $\mathfrak{g}_\pm$. Therefore, the equations
\begin{align*}
&S_1\frac{\p m_\infty}{\p t_n}S_2^{-1}=S_1\frac{\p S_1^{-1}}{\p t_n}+\frac{\p S_2}{\p t_n}S_2^{-1}=-\frac{\p S_1}{\p t_n}S_1^{-1}+\frac{\p S_2}{\p t_n}S_2^{-1}\in \mathfrak{g}_-\oplus\mathfrak{g}_+,\\
&S_1\frac{\p m_\infty}{\p t_n}S_2^{-1}=S_1(\Lambda^n m_\infty)S_2^{-1}=S_1\Lambda^n S_1^{-1}=L_1^n
\end{align*}
lead to $$\frac{\p S_1}{\p t_n}=-(L_1^n)_-S_1,\quad \text{and}\quad  \frac{\p S_2}{\p t_n}=(L_1^n)_+S_2.$$ Here $-$ (respectively $+$) represents the lower triangular part without diagonal (respectively upper triangular part with diagonals) according to Lie algebra splitting. Similarly, one could obtain the result for time parameter $s$ in the same approach. Therefore, from the definition of Lax matrices, it is easy to obtain the Lax representation
\begin{align*}
\frac{\p L_1}{\p t_n}=[L_1,(L_1^n)_-],\quad \frac{\p L_2}{\p s_n}=[L_2,(L_2^n)_+].
\end{align*}

\subsection{Wave Functions and Tau Functions}
Now we'd like to introduce the wave functions and tau functions, to show the Lax representation is equal to the bilinear form of wave functions and tau functions. 

With the notation $\xi(t,z)=\sum_{k=1}^\infty t_kz^k$, let's define the wave functions
\begin{align}\label{wave}
\Phi_1(t,s;z)=e^{\xi(t,z)}S_1\chi(z),\quad \Phi_2(t,s;z)=e^{\xi(s,z^{-1})}S_2\chi(z)
\end{align}
and the dual wave functions 
\begin{align}\label{dualwave}
\Phi_1^*(t,s;z)=e^{-\xi(t,z)}S_1^{-\top}\chi(z^{-1}),\quad \Phi_2^*(t,s;z)=e^{-\xi(s,z^{-1})}S_2^{-\top}\chi(z^{-1}).
\end{align}
It is remarkable that these wave functions satisfy the following bilinear equations \cite{adler97,ueno84}
\begin{align}\label{bltoda}
\oint_{C_\infty} \Phi_{1,n}(t,s;z)\Phi_{1,m}^*(t',s';z)\frac{dz}{2\pi i z}=\oint_{C_0}\Phi_{2,n}(t,s;z)\Phi_{2,m}^*(t',s';z)\frac{dz}{2\pi i z},
\end{align}
where $C_\infty$ and $C_0$ stand for the circles surrounding infinity and zero respectively. Furthermore, if one defines tau functions
by
\begin{align}\label{tau} 
\tau_n(t,s)=\det(m_{i,j}(t,s))_{i,j=0}^{n-1},
\end{align} then one could find the relations between wave functions and tau functions via ($\tilde{\p}_t=(\p_{t_1},\frac{1}{2}\p_{t_2},\cdots)$)
\begin{align*}
&\Phi_1(t,s;z)=\left(e^{\xi(t,z)}\frac{e^{-\xi(\tilde{\p}_t,z^{-1})}\tau_n}{\tau_n}z^n\right)_{n\geq0},\quad \Phi^*_1(t,s;z)=\left(e^{-\xi(t,z)}\frac{e^{\xi(\tilde{\p}_t,z^{-1})}\tau_{n+1}}{\tau_{n+1}}z^{-n}\right)_{n\geq0},\\
&\Phi_2(t,s;z)=\left(e^{\xi(s,z^{-1})}\frac{e^{-\xi(\tilde{\p}_s,z^{})}\tau_{n+1}}{\tau_n}z^n\right)_{n\geq0},\quad
\Phi_2^*(t,s;z)=\left(e^{-\xi(s,z^{-1})}\frac{e^{\xi(\tilde{\p}_s,z)}\tau_n}{\tau_{n+1}}z^{-n}\right)_{n\geq0}.
\end{align*}
If we introduce the notation $[z]=(z,\frac{z^2}{2},\frac{z^3}{3},\cdots)$, then the bilinear identities \eqref{bltoda} could be written into tau-functions as
\begin{align}
\begin{aligned}\label{bilinear}
&\oint_{C_\infty}\tau_n(t-[z^{-1}],s)\tau_{m+1}(t'+[z^{-1}],s')e^{\xi(t-t',z)}z^{n-m-1}dz\\
&=\oint_{C_0}\tau_{n+1}(t,s-[z])\tau_{m}(t',s'+[z])e^{\xi(s-s',z^{-1})}z^{n-m-1}dz.
\end{aligned}
\end{align}

\section{Symmetric reduction and C-Toda hierarchy}\label{sec:symmetric}

In this section, we consider the symmetric reduction to the moment matrix, with our aim being to develop a theory of the symmetric tau function theory analogous to that done in \cite{adler1999,adler02,adler2002} for the skew symmetric case. The symmetric Cholesky decomposition would play an important role in the discussion of this part. Therefore, in subsection \ref{3.1}, we consider the symmetric initial moment matrix and find out that there are some relations between the pairs of dressing operators $S_1$ and $S_2$, wave functions $\Phi$ and $\Phi^*$ and Lax operators $L_1$ and $L_2$. Furthermore, the symmetric features and the commuting flow $\p_{t_n}m_\infty=\Lambda^nm_\infty+m_\infty\Lambda^{\top n}$ result in the relations between the symmetric tau functions and 2d-Toda's tau functions (See Proposition \ref{prop1}). Unfortunately, the relations between symmetric tau functions and 2d-Toda tau functions are nonlinear, thus breaking down the bilinear structures. The failure of bilinear forms ask us for some extra constraints. Motivated by the Cauchy two-matrix model, a rank one shift condition is given. This condition provides an explicit relation between auxiliary functions and tau functions, and thus help us to write down the whole C-Toda hierarchy successfully. The subsection \ref{sub:polynomial} is devoted to a detailed discussion about why the rank one shift condition works and finally demonstrate that the rank one shift condition is equal to a local spectral problem. The Lax integrability of the whole hierarchy is considered in the end.

\subsection{The Symmetric Reduction on 2d-Toda Hierarchy}\label{3.1}
From \eqref{solution}, we know that the most general moment matrix has the form
\begin{align*}
m_\infty(t,s)=e^{\sum_{i=1}^\infty t_i\Lambda^i}m_\infty(0,0)e^{-\sum_{i=1}^\infty s_i\Lambda^{\top i}},
\end{align*}
and two constraints are considered here. The first is to deal with the symmetric initial data $m_\infty(0,0)$. To assume that $m^\top_\infty(0,0)=m_\infty(0,0)$, one could show the relationship between $m_\infty(t,s)$ and $m_\infty(-s,-t)$, and further to show the relations between the corresponding tau functions, Lax operators and wave functions. The second part is to impose the time constraint `$s=-t$', from which one could induce commuting time flows $\p_{t_n}m_\infty=\Lambda^n m_\infty+m_\infty \Lambda^{\top n}$. In this case, tau function depends on time parameters $\{t_k\}_{k=1,2,\cdots}$ only. Some relations between symmetric tau functions and 2d-Toda's tau functions are obtained.

\subsubsection{Symmetric Reduction on Moment Matrix}
Consider the symmetric initial moment matrix $m_\infty(0,0)$ satisfies $m_\infty^\top(0,0)=m_\infty(0,0)$. It is easy to show
$$m_\infty^\top(t,s)=e^{-\sum_{i=1}^\infty s_i\Lambda^i}m^\top(0,0)e^{\sum_{i=1}^\infty t_i\Lambda^{\top i}}=m_\infty(-s,-t).$$ Moreover, the definition of tau function $\tau_n(t,s)=\det(m_{i,j}(t,s))_{i,j=0}^{n-1}$ gives us $\tau_n(-s,-t)=\tau_n(t,s)$, and it leads to
$$h(t,s)=\left(
\frac{\tau_1(t,s)}{\tau_0(t,s)},\frac{\tau_2(t,s)}{\tau_1(t,s)},\cdots\right)=h(-s,-t).$$ Furthermore, we can state the following proposition under this symmetric reduction.

\begin{proposition}\label{symmetricinitial}
For the symmetric initial moment matrix, the following equations are satisfied
\begin{enumerate}
\item $h^{-1}(t,s)S_1(t,s)=S_2^{-T}(-s,-t)$;
\item $h^{-1}(t,s)\Phi_1(t,s;z)=\Phi_2^*(-s,-t;z^{-1})$ and $h^{-1}(t,s)\Phi_2(t,s;z)=\Phi_1^*(-s,-t;z^{-1})$;
\item $L_1(t,s)=h(-s,-t)L_2^\top(-s,-t)h^{-1}(-s,-t)$.
\end{enumerate}
\end{proposition}

\begin{proof}
Starting with $m_\infty^\top(t,s)=m_\infty(-s,-t)$ and Borel decomposition $m_\infty(t,s)=S^{-1}_1(t,s)S_2(t,s)$, we know
\begin{align*}
S_1^{-1}(t,s)S_2(t,s)=S^\top_2(-s,-t)S_1^{-\top}(-s,-t).
\end{align*}
Recall $S_1$ and $S_2$ belong to the groups $\mathfrak{G}_\pm$, therefore, we can rewrite the right hand side as
\begin{align*}
S^\top_2(-s,-t)S_1^{-\top}(-s,-t)=S_2^\top(-s,-t)h^{-1}(-s,-t)h(-s,-t)S_1^{-\top}(-s,-t)
\end{align*}
with $S^\top_2(-s,-t)h^{-1}(-s,-t)\in \mathfrak{G}_-$ and $h(-s,-t)S_1^{-\top}(-s,-t)\in\mathfrak{G}_+$.
Moreover,
from the uniqueness of Borel decomposition, it must satisfy
\begin{align*}
S_1^{-1}(t,s)=S_2^\top(-s,-t)h^{-1}(-s,-t), \quad S_2(t,s)=h(-s,-t)S^{-\top}_1(-s,-t),
\end{align*}
which gives rise to the first equality. 

The second equality is a direct generalisation of the first one if one notices the definition of wave functions \eqref{wave} and their dual \eqref{dualwave} and compute
\begin{align*}
&h^{-1}(t,s)\Phi_1(t,s;z)=h^{-1}(t,s)e^{\xi(t,z)}S_1(t,s)\chi(z)=e^{\xi(t,z)}S_2^{-\top}(-s,-t)\chi(z)=\Phi_2^*(-s,-t;z^{-1}).\\
&h^{-1}(t,s)\Phi_2(t,s;z)=h^{-1}(t,s)e^{\xi(s,z^{-1})}S_2(t,s)\chi(z^{-1})=e^{\xi(s,z^{-1})}S_1^{-\top}\chi(z^{-1})=\Phi_1^*(-s,-t;z^{-1}).
\end{align*}
Finally, a direct computation 
\begin{align*}
L_1(t,s)=S_1(t,s)\Lambda S_1^{-1}(t,s)&=h(-s,-t)S_2^{-\top}(-s,-t)\Lambda S_2^\top(-s,-t)h^{-1}(-s,-t)\\
&=h(-s,-t)L^\top_2(-s,-t)h^{-1}(-s,-t)
\end{align*}
results in the third equation.
\end{proof}

\subsubsection{The ``$s=-t$'' Constraint}
The assumption $s=-t$ shows that the moment matrix depends on parameters $\{t_k\}_{k=1,2,\cdots}$ only, satisfying the time evolutions
\begin{align*}
\p_{t_n}m_{i,j}=m_{i+n,j}+m_{i,j+n},\quad {\text{or}}\quad
\p_{t_n}m_\infty=\Lambda^n m_\infty+m_\infty \Lambda^{\top n}.
\end{align*}
Under this framework, one can define the symmetric tau function $\tilde{\tau}(t)$ by $\tilde{\tau}_n(t):=\tau_n(t,-t)$. To find the integrable hierarchy satisfied by $\{\tilde{\tau}_n(t)\}_{n\in\mathbb{N}}$ with $\tilde{\tau}_0(t)=1$, we need to firstly find the relations between the $\tilde{\tau}_n(t)$ and $\tau_n(t,s)|_{s=-t}$ and the result is shown as the following proposition.
\begin{proposition}\label{prop1}
Regarding the evolutions of the symmetric tau function $\tilde{\tau}_n(t)$ and 2d-Toda's tau function $\tau_n(t,s)|_{s=-t}$, they are connected with each other by
\begin{align}\label{todaandctoda}
\tilde{\tau}_m(t)\tilde{\tau}_m(t+[\alpha]-[\beta])-(\beta-\alpha)^2\tilde{\tau}_{m-1}(t-[\beta])\tilde{\tau}_{m+1}(t+[\alpha])=(\tau_m(t+[\alpha]-[\beta],-t))^2.
\end{align}
\end{proposition}
\begin{proof}
This method used in this proof is almost the same with \cite[Theorem 2.2]{adler2002}, and the slight difference lies in the reduction conditions of the initial matrix. Regarding the bilinear identity satisfied by the 2d-Toda's tau function \eqref{bilinear}, if we take $n=m-1$ and ``$s=-t$'' constraint $s=-t+[\beta]$, $t'=t+[\alpha]-[\beta]$ and $s'=-t-[\alpha]$, then by the use of residue theorem, the right hand side of the bilinear identity \eqref{bilinear} could be computed as
\begin{align*}
\frac{1}{2\pi i}&\oint_{C_\infty} \tau_{m-1}(t-[z^{-1}],s)\tau_{m+1}(t'+[z^{-1}],s')e^{\xi(t-t',z)}\frac{dz}{z^2}\\&=(\beta-\alpha)\tau_{m-1}(t-[\beta],-t+[\beta])\tau_{m+1}(t+[\alpha],-t-[\alpha]).
\end{align*}
It is similar to the left hand side, and one could obtain
\begin{align*}
&\frac{1}{2\pi i}\oint_{C_0}\tau_m(t,s-[z])\tau_{m}(t',s'+[z])e^{\xi(s-s',z^{-1})}\frac{dz}{z^2}\\
&=\frac{1}{\alpha-\beta}\left(\tau_m(t,-t+[\beta]-[\alpha])\tau_{m}(t+[\alpha]-[\beta],-t)-\tau_m(t,-t)\tau_m(t+[\alpha]-[\beta],-t-[\alpha]+[\beta])\right).
\end{align*}
From the definition $\tilde{\tau}_n(t)=\tau_n(t,-t)$ and $\tau_n(t,s)=\tau_n(-s,-t)$, the above two equalities would result in the equation \eqref{todaandctoda}. 
\end{proof}
This proposition provides a way to study the symmetric tau functions from 2d-Toda theory. Taking the square root of both sides, one could obtain\footnote{Here we just take the positive sign when taking the square root of both sides. It doesn't matter if we take the negative sign since it wouldn't affect the forms of the equations.}
\begin{align}\label{symmreduction}
\tau_{m}(t+[\alpha]-[\beta],-t)=\sqrt{\tilde{\tau}_m(t)\tilde{\tau}_m(t+[\alpha]-[\beta])-(\beta-\alpha)^2\tilde{\tau}_{m-1}(t-[\beta])\tilde{\tau}_{m+1}(t+[\alpha])}.
\end{align}
Unlike the Pfaff lattice, the bilinear form \eqref{bilinear} would break down if we replace the tau functions of 2d-Toda by the symmetric tau functions with relationship \eqref{symmreduction}. Nonetheless, this equation implies us some useful information about the $t$-derivatives of symmetric tau functions and 2d-Toda's tau functions, which helps us to find the integrable hierarchy later.
\begin{itemize}
\item If we take the limit $\alpha\to\beta$, then we immediately obtain
\begin{align}\label{zeroorder}
\tau_m(t,s)|_{s=-t}=\tilde{\tau}_m(t);
\end{align}
\item If we take the derivative about $\alpha$ on both sides, then we obtain
\begin{align*}
\sum_{k=1}^\infty \alpha^{k-1}\p_{t_k}\tau_m(t+[\alpha]-[\beta],-t)&=\frac{1}{2}\frac{1}{\tau_m(t+[\alpha]-[\beta],-t)}\\&\times\{\tilde{\tau}_m(t)\times\sum_{k=1}^\infty \alpha^{k-1}\p_{t_k}\tilde{\tau}_m(t+[\alpha]-[\beta])+O\left((\beta-\alpha)\right)\}.
\end{align*}
Taking the limit $\alpha\to\beta$ and comparing the coefficient of $\alpha^{k-1}$, one obtains
\begin{align}\label{firstorder}
\p_{t_k}\tau_m(t,s)|_{s=-t}=\frac{1}{2}\p_{t_k}\tilde{\tau}_m(t);
\end{align}
\item Moreover, if we take the derivative about $\alpha$ and $\beta$ on both sides, and take $\alpha\to\beta$, then we get
\begin{align*}
-2\sum_{k,l=1}^\infty \p_{t_l}&\tau_m(t,-t)\p_{t_k}\tau_m(t,-t)\alpha^{k+l-2}-2\tau_m(t,-t)\sum_{k,l=1}^\infty\p_{t_l}\p_{t_k}\tau_m(t,-t)\alpha^{k+l-2}\\
&=-\tilde{\tau}_m(t)\sum_{k,l=1}^\infty \p_{t_k}\p_{t_l}\tilde{\tau}_m(t)\alpha^{k+l-2}+2\sum_{k,l=0}^\infty p_k(-\tilde{\p}_t)\tilde{\tau}_{m-1}(t)p_l(\tilde{\p}_t)\tilde{\tau}_{m+1}(t)\alpha^{k+l}.
\end{align*}
The special case $\alpha^0$ gives rise to the equation
\begin{align}\label{secondorder}
{\p_{t_1}^2}\tau_m(t,s)|_{s=-t}=\frac{1}{2}{\p_{ t_1}^2}\tilde{\tau}_m(t)-\frac{1}{\tilde{\tau}_m(t)}\left[
{\frac{1}{4}({\p_{t_1}}\tilde{\tau}_m(t))^2+\tilde{\tau}_{m-1}(t)\tilde{\tau}_{m+1}(t)}\right].
\end{align}
\end{itemize}
Most generally, one can state the following proposition.

\begin{corollary}
For arbitrary $i,j\in\mathbb{N}$, there exists the general relation between the symmetric tau function and 2d-Toda's tau functions via
\begin{align}\label{general}
\begin{aligned}
&\sum_{k=0}^i\sum_{l=0}^j p_k(\tilde{\p}_t)p_l(-\tilde{\p}_t)\tau_m(t,s)|_{s=-t}\times p_{i-k}(\tilde{\p}_t)p_{j-l}(-\tilde{\p}_t)\tau_m(t,s)|_{s=-t}\\
&=\tilde{\tau}_m(t)p_i(\tilde{\p}_t)p_j(-\tilde{\p}_t)\tilde{\tau}_m(t)-p_{i-2}(-\tilde{\p}_t)\tilde{\tau}_{m-1}(t)p_j(\tilde{\p}_t)\tilde{\tau}_{m+1}(t)\\
&+2p_{i-1}(-\tilde{\p}_t)\tilde{\tau}_{m-1}(t)p_{j-1}(\tilde{\p}_t)\tilde{\tau}_{m+1}(t)-p_i(-\tilde{\p}_t)\tilde{\tau}_{m-1}(t)p_{j-2}(\tilde{\p}_t)\tilde{\tau}_{m+1}(t).
\end{aligned}
\end{align}
\end{corollary}

\begin{proof}
This equation is from the expansion
\begin{align*}
\tau_m(t+[\alpha]-[\beta],s)|_{s=-t}&=e^{\xi(\tilde{\p}_t,\alpha)-\xi(\tilde{\p}_t,\beta)}\tau_m(t,s)|_{s=-t}\\&=\sum_{k,l=0}^\infty\alpha^k\beta^lp_k(\tilde{\p}_t)p_l(-\tilde{\p}_t){\tau}_m(t,s)|_{s=-t}.
\end{align*}
By comparing with the coefficients of $\alpha^i\beta^j$ on both sides of \eqref{todaandctoda}, we can arrive at this conclusion.
\end{proof}
\begin{remark}
It should be pointed out that equations \eqref{zeroorder}-\eqref{secondorder} could also be obtained from this unified expression \eqref{general} by picking the label $(i,j)$ as $(0,0)$, $(1,0)$ and $(1,1)$, respectively.
\end{remark}

\subsubsection{symmetric moment matrix and Cholesky decomposition}\label{Cholesky} At the end of this subsection, we'd like to mention that these symmetric moment matrices could not only be embedded in the 2d-Toda theory, but they admit the  symmetric Cholesky decomposition. The Cholesky decomposition would lead to an equivalent expression. Let's consider the symmetric matrix $m_\infty(t)$ and the decomposition
\begin{align*}
m_\infty(t)=S^{-1}(t)S^{-\top}(t), \quad S(t)\in\mathcal{G}_-,
\end{align*}
where $\mathcal{G}_-$ is the group composed of the lower triangular matrices with nonzero diagonals. Moreover, the matrix $S(t)$ defines wave operators
\begin{align}\label{waveoperator2}
W_1(t)=S(t)e^{{\sum_{i=1}^\infty}t_i\Lambda^i},\quad W_2(t)=S^{-\top}(t)e^{-\sum_{i=1}^\infty t_i\Lambda^{\top i}},
\end{align}
wave functions (c.f. \eqref{wave} and \eqref{dualwave})
\begin{align*}
\Psi_1(t;z)=W_1(t)\chi(z)=e^{\sum_{i=1}^\infty t_iz^i}S(t)\chi(z),\, \Psi_2(t;z)=W_2(t)\chi(z)=e^{-\sum_{i=1}^\infty t_iz^i}S^{-\top}(t)\chi(z),
\end{align*}
and their dual
\begin{align*}
\Psi_1^*(t;z)=W_1(t)^{-\top}\chi(z^{-1})=\Psi_2(t;z^{-1}),\,\Psi_2^*(t;z)=W_2(t)^{-\top}\chi(z^{-1})=\Psi_1(t;z^{-1})
\end{align*}
with asymptotics 
\begin{align*}
\Psi_{1,n}(t;z)=e^{\sum_{i=1}^\infty t_iz^i} z^nc_n(t)\psi_{1,n}(t;z),\quad \psi_{1,n}=1+O(z^{-1}),\\
\Psi_{2,n}(t;z)=e^{-\sum_{i=1}^\infty t_iz^i}z^n c_n^{-1}(t)\psi_{2,n}(t;z),\quad \psi_{2,n}=1+O(z),
\end{align*}
where the $c_n$ are the diagonals of matrix $S$. Furthermore, the Lax matrix $L_1=S\Lambda S^{-1}$ and $L_2=S^{-\top}\Lambda^\top S^{\top}$ are connected by the simple relation $L_1=L_2^\top$.
From these relations, one can state the following proposition.
\begin{proposition}
For the symmetric moment matrix with commuting vector field $\p_{t_n}m_\infty=\Lambda^n m_\infty+m_\infty\Lambda^{\top n}$, the Cholesky decomposition $m_\infty(t)=S^{-1}(t)S^{-\top}(t)$ gives rise to the wave operators $W_1(t)$ and $W_2(t)$ via \eqref{waveoperator2}. Moreover, these two operators satisfy the relation
\begin{align*}
W_1(t)W_1^{-1}(t')=W_2(t)W^{-1}_2(t'),\quad \forall \, t,\, t'\in \mathbb{C}^\infty,
\end{align*}
which implies the bilinear identity for the wave functions
\begin{align*}
\oint_{C_\infty} \Psi_{1,n}(t;z)\Psi_{2,m}(t';z^{-1})\frac{dz}{2\pi iz}=\oint_{C_0}\Psi_{2,n}(t;z)\Psi_{1,m}(t';z^{-1})\frac{dz}{2\pi iz}.
\end{align*}
\end{proposition}

\subsection{Rank One Shift Condition and Integrable Hierarchy}\label{3.2}
In the last subsection, we demonstrate the bilinear identity loses if we replace the 2d-Toda's tau function by symmetric tau function. In this subsection, we would like to impose another condition for the moment matrix such that the symmetric tau functions could be put into an integrable hierarchy, which is called as the C-Toda hierarchy.

This condition, firstly appeared in \cite{bertola2010} , is called the rank one shift condition. It means the shift of moment matrix satisfies a rank one decomposition
\begin{align}\label{rankone}
\Lambda m_\infty+m_\infty\Lambda^\top=\alpha\alpha^\top,\quad \alpha=(\alpha_0,\alpha_1,\cdots)^\top,\quad \alpha_i(t)\in C^\infty.
\end{align}
It should be pointed out that the time evolution of the vector $\alpha$ should satisfy the condition
\begin{align}\label{atime}
\p_{t_i}\alpha=\Lambda^i \alpha,
\end{align}
to be consistent with that of moment matrix $m_\infty$.
Then according to the Borel decomposition $m_\infty=S_1^{-1}S_2$, the rank one condition is equal to
\begin{align*}
\Lambda S_1^{-1}S_2+S_{1}^{-1}S_2\Lambda^\top=\alpha\alpha^\top\quad S_1\Lambda S_1^{-1}+S_2\Lambda^\top S_2^{-1}=S_1\alpha\alpha^\top S_2^{-1}.
\end{align*}
From the Proposition \ref{symmetricinitial} and condition ``$s=-t$'', we know $S_2=hS_1^{-\top}$ and $L_2=hL_1^\top h^{-1}$. Therefore, the above equality could then be written in a symmetric form
\begin{align}\label{rankonecondition}
L_1h+hL_1^{\top}=(S_1\alpha)(S_1\alpha)^\top:=\sigma\sigma^\top
\end{align}
with $\sigma=S_1\alpha$.
In addition, the rank one shift condition \eqref{rankone}, or its equivalent condition \eqref{rankonecondition}, implies that there are some relations between symmetric $\tau$-functions and auxiliary functions $\{\sigma_n\}_{n\geq0}$.

\begin{proposition}\label{prop:3.4}
If we denote $L_1$ as
\begin{align*}
L_1=\left(
\begin{array}{ccccc}
l_{00}&1&&&\\
l_{10}&l_{11}&1&&\\
l_{20}&l_{21}&l_{22}&1&\\
\vdots&\vdots&\vdots&\vdots&\ddots
\end{array}
\right):=(l_{i,j}),
\end{align*}
then there are relations between tau functions and $\sigma$ via
\begin{subequations}
\begin{align}
&2l_{j,j}h_j=\sigma_j\sigma_j,\label{ctoda1}\\
&h_{j+1}+l_{j+1,j}h_j=\sigma_j\sigma_{j+1},\label{ctoda2}\\
&l_{i,j}h_j=\sigma_i\sigma_j,\, i\geq j+2.\label{ctoda3}
\end{align}
\end{subequations}
\end{proposition}

Now we plan to show the equations \eqref{ctoda1}-\eqref{ctoda3} are the hierarchy of intergable lattice, whose first example is the C-Toda lattice given by \cite{li2019}.
By the use of the fact $L_1=S_1\Lambda S_1^{-1}$ and denote $S_1=(s_{i,j})$, $S_1^{-1}=(\tilde{s}_{i,j})$, we have 
\begin{align*}
L_1=\left(
\begin{array}{cccc}
1&&&\\
s_{10}&1&&\\
s_{20}&s_{21}&1&\\
\vdots&\vdots&\vdots&\ddots\end{array}
\right)\left(
\begin{array}{ccccc}
0&1&&&\\
&0&1&&\\
&&0&1&\\
&&&\ddots&\ddots
\end{array}
\right)\left(
\begin{array}{cccc}
1&&&\\
\tilde{s}_{10}&1&&\\
\tilde{s}_{20}&\tilde{s}_{21}&1&\\
\vdots&\vdots&\vdots&\ddots\end{array}
\right)
\end{align*}
which shows 
\begin{align}\label{lij}
l_{i,j}=\tilde{s}_{i+1,j}+s_{i,j-1}+\sum_{k=j}^{i-1}s_{i,k}\tilde{s}_{k+1,j}.
\end{align}
Moreover, since $S_1^{-1}$ is the inverse matrix of $S_1$, and notice $S_1$ and $S_1^{-1}$ are strictly lower triangular matrix with diagonals $1$, we have
\begin{align}\label{stilde}
\tilde{s}_{i,j}=-s_{i,j}-\sum_{k=j+1}^{i-1}s_{i,k}\tilde{s}_{k,j},
\end{align}
In addition, from the 2d-Toda theory, we know $$s_{i,j}=\frac{p_{i-j}(-\tilde{\p}_t)\tau_i(t,s)}{\tau_i(t,s)}.$$ Therefore, we have the following proposition.
\begin{proposition}
The main diagonal and sub-diagonal give rise to the C-Toda lattice
\begin{subequations}
\begin{align}
&D_{t_1}\tilde{\tau}_{n+1}\cdot\tilde{\tau}_n=\tilde{\sigma}_n\tilde{\sigma}_n,\label{c1}\\
&D_{t_1}^2\tilde{\tau}_{n+1}\cdot\tilde{\tau}_{n+1}=4\tilde{\sigma}_n\tilde{\sigma}_{n+1},\label{c2}
\end{align}
\end{subequations}
where $\tilde{\sigma}_n=\tilde{\tau}_n\sigma_n$ and $D_t$ is the standard Hirota operator defined by
\begin{align*}
D_t^m f(t)\cdot g(t)=\frac{\partial^m}{\partial x^m}f(t+x)g(t-x)|_{x=0}.
\end{align*}
\end{proposition}

\begin{proof}

Firstly, we prove the equality \eqref{c1}. From equation \eqref{lij} and \eqref{stilde}, one can demonstrate the elements of main diagonal can be written as
 $$l_{j,j}=s_{j,j-1}+\tilde{s}_{j+1,j}=s_{j,j-1}-s_{j+1,j},$$
and as is indicated, $$s_{j,j-1}=\frac{p_1(-\tilde{\p}_t){\tau}_j(t,s)|_{s=-t}}{{\tau}_j(t,s)|_{s=-t}}.$$ 
Furthermore, from the relations  \eqref{zeroorder} and \eqref{firstorder}, one can express
$s_{j,j-1}$ in terms of symmetric tau function as $s_{j,j-1}=-\frac{1}{2}\p_{t_1}\log\tilde{\tau}_j(t)$. Therefore, the main diagonal \eqref{ctoda1} leads to 
\begin{align*}
\p_{t_1}\log\frac{\tilde{\tau}_{j+1}}{\tilde{\tau}_j}\times\frac{\tilde{\tau}_{j+1}}{\tilde{\tau}_j}=\frac{\tilde{\sigma}_j\tilde{\sigma}_j}{\tilde{\tau}_j^2},
\end{align*}
which is indeed \eqref{c1}. 

Now we turn to equation \eqref{c2}. From the equations
\begin{align*}
&l_{j+1,j}=s_{j+1,j-1}+s_{j+1,j}\tilde{s}_{j+1,j}+\tilde{s}_{j+2,j}=s_{j+1,j-1}-s_{j+2,j}+s_{j+2,j+1}s_{j+1,j}-s_{j+1,j}^2,\\
&s_{j+1,j-1}=\frac{p_2(-\tilde{\p}_t){\tau}_{j+1}(t,s)|_{s=-t}}{{\tau}_{j+1}(t,s)|_{s=-t}},
\end{align*} one knows
\begin{align*}
s_{j+1,j-1}-s_{j+2,j}=\frac{1}{2}\p_{t_2}\log\frac{{\tau}_{j+2}(t,s)|_{s=-t}}{{\tau}_{j+1}(t,s)|_{s=-t}}+\frac{1}{2}\frac{\p_{t_1}^2{\tau}_{j+1}(t,s)|_{s=-t}}{{\tau}_{j+1}(t,s)|_{s=-t}}-\frac{1}{2}\frac{\p_{t_1}^2{\tau}_{j+2}(t,s)|_{s=-t}}{{\tau}_{j+2}(t,s)|_{s=-t}}.
\end{align*}
It seems that $t_2$-flow would involve in this equation. Interestingly, the $t_2$-flow can be cancelled from the 2d-Toda hierarchy since $\{{\tau}_n(t,s)|_{s=-t}\}_{n=0}^\infty$ satisfy the modified KP hierarchy \cite{jimbo1983}, whose first equation is $(D_{t_1}^2+D_{t_2}){\tau}_{j+1}(t,s)\cdot{\tau}_{j+2}(t,s)|_{s=-t}=0$. This equation implies
\begin{align*}
\p_{t_2}\log\frac{{\tau}_{j+2}(t,s)|_{s=-t}}{{\tau}_{j+1}(t,s)|_{s=-t}}&=-2\frac{\p_{t_1}\tau_{j+1}(t,s)|_{s=-t}\p_{t_1}\tau_{j+2}(t,s)|_{s=-t}}{\tau_{j+1}(t,s)|_{s=-t}\tau_{j+2}(t,s)|_{s=-t}}\\
&+\frac{\p_{t_1}^2{\tau}_{j+1}(t,s)|_{s=-t}}{{\tau}_{j+1}(t,s)|_{s=-t}}+\frac{\p_{t_1}^2{\tau}_{j+2}(t,s)|_{s=-t}}{{\tau}_{j+2}(t,s)|_{s=-t}},
\end{align*}
and therefore
\begin{align}\label{l3}
l_{j+1,j}=\frac{{\tau}_{j+1}(t,s)|_{s=-t}\p_{t_1}^2{\tau}_{j+1}(t,s)|_{s=-t}-(\p_{t_1}{\tau}_{j+1}(t,s)|_{s=-t})^2}{({\tau}_{j+1}(t,s)|_{s=-t})^2}.
\end{align}
By using the derivatives among $\tilde{\tau}_j(t)$ and ${\tau}_j(t,s)|_{s=-t}$ in \eqref{zeroorder}-\eqref{secondorder},  we finally arrive at
\begin{align*}
l_{j+1,j}=\frac{\frac{1}{4}D_{t_1}^2\tilde{\tau}_{j+1}\cdot\tilde{\tau}_{j+1}}{\tilde{\tau}_{j+1}^2}-\frac{\tilde{\tau}_j\tilde{\tau}_{j+2}}{\tilde{\tau}_{j+1}^2}
\end{align*}
By taking $l_{j+1,j}$ into equation \eqref{ctoda2}, the second equation of C-Toda lattice can be obtained.
\end{proof}

One should notice: the rank one shift condition is one of the necessary conditions to find these quadralinear equations, and the other is the symmetric reduction on the $\tau$-function of 2d-Toda hierarchy. These two conditions, could be regarded as the spectral problem and time evolutions of the C-Toda hierarchy respectively, which are the inherent structures of integrability and we'd like to give an explanation in the subsection \ref{laxpair}.

\subsection{Rank One Shift Condition Revisited: Polynomial Case}\label{sub:polynomial}
In fact, the rank one shift condition has a natural explanation in polynomial sense. If we denote the symmetric moment matrix $m_\infty=(m_{i,j})_{i,j\in\mathbb{N}}$, composed by moments generated by Cauchy kernel as
\begin{align}\label{cbopmoment}
m_{i,j}=\int_{\mathbb{R}^2_+}\frac{x^iy^j}{x+y}\omega(x)\omega(y)dxdy,
\end{align}
then from the rank one shift condition \eqref{rankone},
one knows $\{\alpha_i\}_{i\in\mathbb{Z}}$ can be expressed as
\begin{align}\label{alpha}\alpha_i=\int_{\mathbb{R}_+} x^i\omega(x)dx.
\end{align}
In this case, we always assume that the weight function $\omega(x)$ decays fast when $x\to\infty$ so that the moments are nonsingular. As a way of illustration, in \cite{bertola2014}, $\omega(x)$ is taken as the Laguerre weight, and the corresponding Cauchy-Laguerre polynomials are considered, with the asymptotic behaviour depicted by Meijer G-function.

At the same time, these moments in \eqref{cbopmoment} could be interpreted by the symmetric bilinear inner product $\langle \cdot,\cdot\rangle$ mapping $\mathbb{R}[x]\times\mathbb{R}[y]\mapsto\mathbb{R}$, such that
\begin{align}\label{innerproduct}
\langle x^i,y^j\rangle:=m_{i,j}=\int_{\mathbb{R}^2_+}\frac{x^iy^j}{x+y}\omega(x)\omega(y)dxdy.
\end{align}
Under the bilinear inner product, one can introduce a family of polynomials $\{P_n(x)\}_{n=0}^\infty$, satisfying the orthogonal relation 
\begin{align*}
\langle P_n(x),P_m(y)\rangle=h_n\delta_{m,n},\quad h_n=\frac{\det(m_{i,j})_{i,j=0}^n}{\det(m_{i,j})_{i,j=0}^{n-1}}.
\end{align*}
We call this family of polynomials the symmetric Cauchy bi-orthogonal polynomials since $m_{i,j}=m_{j,i}$ and a determinant expression of these polynomials could be found in \cite{li2019}. 

\begin{remark}
In this paper, we mainly consider the symmetric reduction of moment matrix, therefore from the Cholesky decomposition we know the dressing operator can be chosen as the same upon a constant factor. Therefore, the two family of Cauchy bi-orthogonal polynomials in this case is reduced to one family of polynomials. This is a special case of Cauchy bi-orthogonal polynomials motivated by the structure of symmetric tau functions.
\end{remark}

To be consistent with the time evolution \eqref{atime}, from now on, we assume $\omega(x)$ is dependent on time parameters $\{t_k\}_{k=1,2,\cdots}$ with evolutions $\partial_{t_k}\omega(x;t)=x^k\omega(x;t)$.
Furthermore, it should be pointed out that if the time-dependent moments are chosen as
\begin{align*}
m_{i,j}=\int_{\mathbb{R}_+^2}\frac{x^iy^j}{x+y}\omega(x;t)\omega(y;t)dxdy,
\end{align*}
then from the definition of tau function, one knows
\begin{align*}
{\tau}_n=\det(m_{i,j})_{i,j=0}^{n-1}=\det\left(\int_{\mathbb{R}_+^2}\frac{x^iy^j}{x+y}\omega(x;t)\omega(y;t)dxdy\right)_{i,j=0}^{n-1}.
\end{align*}
By using the Andr\'eief formula, one can show
\begin{align*}
\tau_n&=\sum_{\sigma\in S_n}\epsilon(\sigma){\prod_{j=1}^n}\int_{\mathbb{R}_+^2}x_j^{\sigma_j-1}y_j^{j-1}\frac{1}{x_j+y_j}\omega(x_j;t)\omega(y_j;t)dx_jdy_j\\
&=\int_{\mathbb{R}_+^{n}\times \mathbb{R}_+^n}\Delta(x)\prod_{j=1}^n y_j^{j-1}\frac{1}{x_j+y_j}\omega(x_j;t)\omega(y_j;t)dx_jdy_j\\
&=\frac{1}{n!}\sum_{\sigma\in S_n}\epsilon(\sigma)\int_{\mathbb{R}_+^n\times\mathbb{R}_+^n}\Delta(x)\prod_{j=1}^ny_{\sigma_j}^{j-1}\frac{1}{x_j+y_j}\omega(x_j;t)\omega(y_j;t)dx_jdy_j\\
&=\frac{1}{(n!)^2}\int_{\mathbb{R}_+^n\times\mathbb{R}_+^n}\det\left[\frac{1}{x_i+y_j}\right]_{i,j=1}^n\Delta(x)\Delta(y)\prod_{i=1}^n\omega(x_i;t)\omega(y_i;t)dx_idy_i.
\end{align*}
which is the time-dependent partition function of Cauchy two-matrix model.
Moreover, the Cauchy bi-orthogonal polynomials could be written in terms of dressing operator.
\begin{proposition}
We have 
\begin{align*}
P_n(x)=(S_1\chi(x))_n,
\end{align*}
where $(S_1\chi(x))_n$ is the $(n+1)$-th component of the vector $S_1\chi(x)$. 
\end{proposition}

With this expression, one can show the explicit expression for the $\sigma$ as
$$\sigma=S_1\alpha=\left(\int_{\mathbb{R}_+}P_n(x)\omega(x)dx\right)_{n\geq0}$$
if equation \eqref{alpha} is taken into account.
Now that $\sigma$ is a column vector, one can write it as a diagonal matrix acting on a vector composed by $1$, denoting as $\sigma=D_\sigma\mathbbm{1}$. Therefore, we have the following proposition.
\begin{proposition}
The equation \eqref{rankonecondition} is equivalent to
\begin{align*}
(\Lambda-Id)D_\sigma^{-1}L_1h+(\Lambda-Id)D_\sigma^{-1}hL_1^\top=0.
\end{align*}
\end{proposition}
\begin{proof}
Since $\mathbbm{1}$ is the null vector of $(\Lambda-Id)$, then one can immediately obtain this equation by acting $(\Lambda-Id)D_\sigma^{-1}$ on both sides of equation \eqref{rankonecondition}.
\end{proof}
Moreover, if we denote a matrix-value operator $$\mathcal{A}:=(\Lambda-Id)D_\sigma^{-1}L_1h,$$ then $\mathcal{A}\in M_{[-1,2]}$, where $M_{[-1,2]}$ is the set of band matrices with one sub-diagonals and two off-diagonals. It is because $L_1=S_1\Lambda S_1^{-1}$ is a lower Hessenberg matrix, we know $(\Lambda-Id)D_\sigma^{-1}L_1h$ has two off-diagonals while $(\Lambda-Id)D_\sigma^{-1}hL_1^\top$ is still an upper Hessenberg matrix, from which we could conclude $(\Lambda-Id)D_\sigma^{-1}L_1h$ belongs to $M_{[-1,2]}$. This fact implies a four term recurrence relation or a $3\times 3$ spectral problem.
\begin{corollary}\label{ftrrcbop}
The spectral problem $L_1\Phi=z\Phi$ could be decomposed into
\begin{align}\label{spec}
z(\Lambda-Id)D_\sigma^{-1}\Phi=\mathcal{A}h^{-1}\Phi.
\end{align}
\end{corollary}

As has been indicated in 2d-Toda theory, the Lax matrix is usually a nonlocal upper Hessenberg matrix and is difficult to analyse, while under the rank one shift condition, this spectral problem could be decomposed into a local $3\times 3$ spectral problem and some further explanations would be made in the sequent part.

\subsection{Lax Integrability of the C-Toda Lattice}\label{laxpair}
It has been shown that the rank one shift condition gives rise to a local spectral problem and the involvement of time parameters indicates an evolutionary equation. In this part, we'd like to go further with the Lax representation for the C-Toda hierarchy. Under the constraint ``$s=-t$'', the moment matrix $m_\infty$ now is dependent on time variables $\{t_k\}_{k=1,2,\cdots}$ and it satisfies the time evolution
\begin{align}\label{timeevolution}
\p_{t_k}m_\infty=\Lambda^k m_\infty+m_\infty\Lambda^{\top k}.
\end{align}
From the time-dependent Borel decomposition $m_\infty(t)=S_1^{-1}(t)S_2(t)$ and define $L_1(t)=S_1(t)\Lambda S_1^{-1}(t)$, $L_2(t)=S_2(t)\Lambda^\top S_2^{-1}(t)$, we know
\begin{align*}
S_1\frac{\p m_\infty}{\p t_n}S_2^{-1}=L_1^n+L_2^n=-\frac{\p S_1}{\p t_n}S_1^{-1}+\frac{\p S_2}{\p t_n}S_2^{-1}\in\mathfrak{g}_-\oplus\mathfrak{g}_+,
\end{align*}
where $\mathfrak{g}_-$ (respectively $\mathfrak{g}_+$) is the Lie algebra of strictly lower triangular matrices (resp. Lie algebra of upper triangular matrices).
Therefore, according to the Lie algebra splitting, one could obtain
\begin{align}\label{lax1}
\frac{\p S_1}{\p t_n}=-(L_1^n+L_2^n)_-S_1, \quad \frac{\p L_1}{\p t_n}=[-(L_1^n+L_2^n)_-,L_1].
\end{align}
As has been shown in Proposition \ref{symmetricinitial}, we know $S_1^{-1}(t)=S_2^\top(t)h^{-1}(t)$. This fact leads to the relation $L_1=hL_2^\top h^{-1}$, where $h$ is a diagonal matrix admitting the form $h=\text{diag}\{h_0,h_1,\cdots\}$ with $h_n={\tau_{n+1}}/{\tau_n}$.

 Therefore, the Lax equation \eqref{lax1} could be re-expressed as 
\begin{align}\label{lax}
\frac{\p L_1}{\p t_n}=[-(L_1^n+hL_1^{\top n}h^{-1})_-,L_1],
\end{align}
dependent on Lax operator $L_1$ and a diagonal matrix $h$.
Now, we try to give an explicit expression for the Lax operator $L_1$ in terms of nonlinear variables (or symmetric tau-functions introduced before).
With the help of bilinear inner product \eqref{innerproduct}, time evolution \eqref{timeevolution} and rank one shift condition \eqref{rankonecondition}, we could obtain the following four-term recurrence relation and time evolution for the Cauchy bi-orthogonal polynomials \cite[Equation 2.4, Equation 3.2]{li2019}
\begin{align}\label{comp}
\begin{aligned}
&x(P_{n+1}(x;t)+a_nP_n(x;t))=P_{n+2}(x;t)+b_nP_{n+1}(x;t)+c_nP_n(x;t)+d_nP_{n-1}(x;t),\\
&\p_{t_1}P_{n+1}(x;t)+a_n\p_{t_1}P_n(x;t)=a_n\p_{t_1}(\log h_n)P_n(x;t)
\end{aligned}
\end{align}
with the coefficients
\begin{align*}
a_n=-\frac{\sigma_{n+1}\tau_n}{\sigma_n\tau_{n+1}}, \quad b_n=a_n+\frac{1}{2}\p_{t_1}\log h_{n+1},\quad c_n=-\frac{h_{n+1}}{h_n}-\frac{1}{2}a_n\p_{t_1}\log h_n,\quad d_n=-a_n\frac{h_n}{h_{n-1}}.
\end{align*}
If we denote $a=\text{diag}(a_1,a_2,\cdots)$, $b=\text{diag}(b_0,b_1,\cdots)$, $c=\text{diag}(c_1,c_2,\cdots)$ and $d=\text{diag}(d_2,d_3,\cdots)$, then the above recurrence relation and time evolution could be written as
\begin{subequations}
\begin{align}
&(\Lambda^0+\Lambda^{-1}a)xP(x;t)=(\Lambda^1+\Lambda^0b+\Lambda^{-1}c+\Lambda^{-2}d)P(x;t),\label{ftr}\\
&(\Lambda^0+\Lambda^{-1}a)\p_{t_1}P(x;t)=(\Lambda^{-1}e)P(x;t)
\end{align}
\end{subequations}
with $e=\text{diag}(e_1,e_2,\cdots)$ and $e_n=2(b_{n-1}-a_{n-1})a_n$. Here the equation \eqref{ftr} is another equivalent form of Proposition \ref{ftrrcbop}. From the explicit expressions for the symmetric Cauchy biorthogonal polynomials, we can state the following proposition.

\begin{proposition}
The Lax matrix of C-Toda lattice admits the expression 
\begin{align*}
L_1=(\Lambda^0+\Lambda^{-1}a)^{-1}(\Lambda^1+\Lambda^0b+\Lambda^{-1}c+\Lambda^{-2}d).
\end{align*}
Moreover, there exists
\begin{align*}
(\Lambda^0+\Lambda^{-1}a)^{-1}(\Lambda^{-1}e)=-(L_1+hL_1^\top h^{-1})_-.
\end{align*}
\end{proposition}
\begin{proof}
The first equation is obvious since the Cauchy bi-orthogonal polynomials can be viewed as a modification of the wave function and thus the spectral operator $L_1$ can be characterised by the four term recurrence relation. The only thing we need to verify is to show the second equality, i.e.
\begin{align*}
\{(\Lambda^0+\Lambda^{-1}a)^{-1}(\Lambda^1+\Lambda^0b+\Lambda^{-1}c+\Lambda^{-2}d)+h\Lambda^{-1}h^{-1}\}_{\Lambda^{<0}}=-(\Lambda^0+\Lambda^{-1}a)^{-1}(\Lambda^{-1}e)
\end{align*}
after taking the exact value of $L_1$.
Equally, we plan to show
\begin{align*}
\left\{\sum_{j=0}^\infty (-\Lambda^{-1}a)^j(\Lambda^1+\Lambda^0b+\Lambda^{-1}c+\Lambda^{-2}d)+h\Lambda^{-1} h^{-1}\right\}_{\Lambda^{<0}}=-\sum_{j=0}^\infty (-\Lambda^{-1}a)^j(\Lambda^{-1}e).
\end{align*}
For the $\Lambda^{-1}$ term, it is to verify
\begin{align*}
&-\Lambda^{-1}e=\Lambda^{-1}a\Lambda^{-1}a\Lambda^1-\Lambda^{-1}a\Lambda^0b+\Lambda^{-1}c+h\Lambda^{-1}h^{-1},
\end{align*}
and it is equal to verify 
\begin{align*}
 -e_{n-1}=a_{n-1}a_{n-2}-a_{n-1}b_{n-2}+c_{n-1}+h_nh_{n-1}^{-1}.
\end{align*}
By taking the expressions of the coefficients in \eqref{comp}, it is easy to find the equation is valid.

For the reminding parts, firstly, notice that for $j\geq1$,
\begin{align*}
-(-\Lambda^{-1}a)^j\Lambda^{-1}e=(-\Lambda^{-1}a)^{j+2}\Lambda^1+(-\Lambda^{-1}a)^{j+1}\Lambda^0b+(-\Lambda^{-1}a)^j\Lambda^{-1}c+(-\Lambda^{-1}a)^{j-1}\Lambda^{-2}d
\end{align*}
could be expressed as
\begin{align*}
a_{n-1}e_{n-2}=-a_{n-1}a_{n-2}a_{n-3}+a_{n-1}a_{n-2}b_{n-3}-a_{n-1}c_{n-2}+d_{n-1}
\end{align*}
 if we make a cancellation of the highest order term and make a shift of the index.
 Moreover, this equation could be verified if the values of the coefficients are taken, and thus complete this proof.
\end{proof}

\begin{remark}
This hierarchy is interesting in two folds. One is in the classical integrable system theory. As is shown in \cite{chang18}, the first member of the hierarchy could be viewed as the continuum limit of the full discrete CKP equation, and therefore, this hierarchy in some sense could be viewed as a discrete CKP hierarchy. The second is pointed out in \cite{zuo19}, in which the Lax matrix of the form
\begin{align*}
L=(\Lambda-a_{l+2})^{-1}(\Lambda^{k+1}+a_1\Lambda^k+\cdots+a_{l+1}\Lambda^{k-l}),\quad 1\leq k<l
\end{align*}
plays an important role when studying the Frobenius manifold isomorphic to the orbit space of the extended affine Weyl group $\tilde{W}^{(k,k+1)}(A_l)$. One can see when $k=1$ and $l=2$, this Lax matrix implies a four term recurrence relationship of bi-orthogonal polynomials and Cauchy bi-orthogonal polynomials provides a reasonable example in this case.
\end{remark}

At the end of this section, we would like to remark that the Cholesky decomposition discussed in Section \ref{Cholesky} could result in an equivalent Lax representation for this hierarchy. As before, consider $m_\infty=S^{-1}S^{-\top}$ with $S$ the lower triangular matrix with nonzero diagonals. By defining the Lax matrix $L=S\Lambda S^{-1}$ and from the commuting vector field \eqref{timeevolution}, one could obtain
\begin{align*}
-\frac{\p S}{\p t_n}S^{-1}-\frac{\p S^\top}{\p t_n}S^{-\top}=L^n+L^{\top n}.
\end{align*}
Note that $S$ and $S^\top$ admit the same diagonals, therefore, from the Lie algebra decomposition, one has
\begin{align*}
\frac{\p S}{\p t_n}=-[(L^n)_-+(L^{\top n})_{-}+(L^n)_0]S,
\end{align*}
where $-$ and $0$ stand for the strictly lower triangular and diagonal part of the matrix respectively.

\section{Skew symmetric reduction and B-Toda hierarchy}\label{sec:skewsymmetric}
In this section, we would like to consider skew symmetric reduction to the moment matrix, comparing with the symmetric reduction given in the last section. In Subsection \ref{4.1}, we would like to give a review to the skew symmetric reduction and restate the basic ideas about the skew Borel decomposition and the concept of skew orthogonal polynomials. The previous results are important in two aspects. One is that the Pfaff lattice hierarchy could be embedded into 2d-Toda hierarchy and the second is due to the explicit connection between Dyson's unitary matrix model and the symmetric, symplectic random ensembles. The latter demonstrates the relationship between the orthogonal polynomials and skew orthogonal polynomials, which are indeed the wave functions of Pfaff lattice and Toda lattice \cite{adler00,adler02}. 

However, there are some problems still remained in the skew symmetric reductions. In \cite{kac97}, Kac and de van Leur have found a large BKP hierarchy and demonstrated later in \cite{vandeleur01} that the Pfaff lattice hierarchy should be a subclass of the large BKP hierarchy with only even indexed tau functions. In \cite{adler2002}, Adler and van Moerbeke left the exact relation between Pfaff lattice and large BKP hierarchy as an ongoing problem. 
Very recently, Chang et al showed that if the odd indexed Pfaffian tau functions are introduced, some integrable lattices could be derived with the concept of partial skew orthogonal polynomials \cite{chang182}. Therefore, in Subsection \ref{4.2}, we give some new properties about partial skew orthogonal polynomials, from which  an alternative expressions for the odd indexed wave functions for Pfaff lattice hierarchy would be provided. With the help of the odd indexed wave function, we prove that the large BKP hierarchy and the Pfaff lattice hierarchy are exactly the same in the sense of wave functions/tau functions. The second interesting problem is the tau function of the Pfaff lattice hierarchy/large BKP hierarchy are not exactly the tau functions of BKP hierarchy in Hirota's sense since the tau functions of BKP hierarchy should satisfy the derivative law of Gram-type Pfaffian, which motivates us to consider the rank two shift condition here. In Subsection \ref{4.3}, we consider the rank two shift condition for the skew symmetric case in details and embed them into 2d-Toda theory. It is shown that the B-Toda lattice is the first nontrivial example in this case which is not listed in the large BKP hierarchy. Therefore, we make an effort in the rank two shift condition and finally give a novel sub-hierarchy under this condition. A recurrence relation for the particular partial skew orthogonal polynomials is given in terms of the Fay identity and some discussions about the Lax integrability are made at the end of this section.

\subsection{The Skew Symmetric Reduction on the 2d-Toda Hierarchy}\label{4.1}
Consider a skew symmetric moment matrix $m_{2k}$ of the block form
\begin{align*}
m_{2k}=\left[\begin{array}{cc}
M_1&-C^{\top}\\
C&M_2
\end{array}
\right],
\end{align*} 
where $M_1$ is a nonzero $2\times 2$ skew symmetric matrix and $M_2$ is a $(2k-2)\times (2k-2)$ skew symmetric matrix and $C$ is a $(2k-2)\times 2$ matrix. By applying the decomposition
\begin{align*}
\left[\begin{array}{cc}
M_1&-C^\top\\
C&M_2
\end{array}
\right]=\left[\begin{array}{cc}
I&\\
CM_1^{-1}&I
\end{array}
\right]\left[\begin{array}{cc}
M_1&\\
&M_2+CM_1^{-1}C^\top
\end{array}
\right]\left[\begin{array}{cc}
I&-M_1^{-1}C^\top\\
&I
\end{array}
\right],
\end{align*}
and iterating the process, one can see the skew symmetric moment matrix with even order can be decomposed into the form $m_\infty=S^{-1}hS^{-\top}$, where
\begin{align*}
S\in\{A|\,\text{$A$ is a strictly lower triangular matrix in the sense of $2\times2$ block matrix}\},
\end{align*}
\begin{align}\label{sop}
h=\text{diag}\left\{\left(\begin{array}{cc}
0&h_0\\
-h_0&0\end{array}
\right),\left(\begin{array}{cc}
0&h_1\\
-h_1&0\end{array}
\right),\cdots\right\},\quad h_n=\frac{\hat{\tau}_{2n+2}}{\hat{\tau}_{2n}},\quad \hat{\tau}_{2n}=\Pf(m_{i,j})_{i,j=0}^{2n-1}
\end{align}
with the notation $\hat{\tau}_0=1$.

\begin{remark}
In fact, this kind of decomposition has been considered in \cite{adler1999,adler2002} while they put more emphasis on the standard decomposition $m_\infty=\hat{S}^{-1}J\hat{S}^{-\top}$, where $J$ is the matrix of the form
\begin{align*}
J=\text{diag}\left(
J_{2\times 2},\, J_{2\times 2},\cdots
\right),\quad J_{2\times 2}=\left(\begin{array}{cc}
0&1\\
-1&0
\end{array}\right),
\end{align*}
and $\hat{S}$ is the lower triangular block matrix with nonzero diagonals. This standard decomposition is called the skew Borel decomposition and corresponds to the Lie algebra splitting $gl_\infty=k\oplus sp(\infty)$, where $k$ is the Lie algebra of lower triangular matrix with diagonals (in the sense of $2\times 2$ block matrix) and $sp(\infty)$ is the Lie algebra of Symplectic group.
\end{remark}

Two main results have already been known from this skew symmetric LU decomposition. One is to construct skew orthogonal polynomials and corresponding Christoffel-Darboux kernel \cite{adler1999,forrester19}. Starting from the dressing operator $S$, one can set a family of polynomials $\{P_n(x)\}_{n\geq0}$ such that 
\begin{align}\label{sop2}
\begin{aligned}
&(1)\, P(x)=S\chi(x);\\
&(2)\, \langle P(x),P^\top(y)\rangle=h,
\end{aligned}
\end{align}
where $P(x)$ is the vector of skew orthogonal polynomials $P(x)=(P_0(x),P_1(x),\cdots)^\top$, $\langle \cdot,\cdot\rangle$ is a skew symmetric inner product defined by $\langle x^i,y^j\rangle=m_{i,j}$ such that $m_{i,j}=-m_{j,i}$ and $h$ is given in \eqref{sop}. Moreover, this family of skew orthogonal polynomials could be written in explicit Pfaffian expressions \cite{adler1999,chang182,forrester19}
\begin{align}\label{sops}
P_{2n}(x)=\frac{1}{\htau_{2n}}\Pf(0,\cdots,2n,x),\quad P_{2n+1}(x)=\frac{1}{\htau_{2n}}\Pf(0,\cdots,2n-1,2n+1,x),
\end{align}
where $\htau_{2n}=\Pf(0,\cdots,2n-1)$, $\Pf(i,j)=m_{i,j}$ and $\Pf(i,x)=x^i$. 

The other is to connect with integrable hierarchy by the introduction of evolutions \cite{adler1999,adler2002}. Let's consider the moment matrix admitting the form
\begin{align*}
m_\infty(t,s)=e^{\sum_{i=1}^\infty t_i\Lambda^i}m_\infty(0,0)e^{-\sum_{i=1}^\infty s_i\Lambda^{\top i}}
\end{align*}
with initial moment matrix $m^\top_\infty(0,0)=-m_\infty(0,0)$. The constraint ``$s=-t$'' shows us the commuting vector field
\begin{align}\label{tes}
\p_{t_n} m_\infty=\Lambda^n m_\infty+m_\infty\Lambda^{\top n},
\end{align}
and therefore the skew orthogonal polynomials have the relation
\begin{align}\label{rel1}
(z+\p_{t_1})(\htau_{2n}P_{2n}(x;t))=\htau_{2n}P_{2n+1}(x;t),
\end{align}
which is a mixture of the spectral problem and time evolution.
Moreover, the skew symmetric tau functions $\{\htau_{2n}(t)\}_{n\geq0}$ defined in \eqref{sop} have the following connections with 2d-Toda's tau functions $\{\tau_n(t,s)\}_{n\geq0}$ via relations (see \cite[Theorem 2.2]{adler2002})
\begin{align*}
\tau_{2n}(t+[\alpha]-[\beta],-t)&=\htau_{2n}(t)\htau_{2n}(t+[\alpha]-[\beta]), \\ 
\tau_{2n+1}(t+[\alpha]-[\beta],-t)&=(\alpha-\beta)\htau_{2n}(t-[\beta])\htau_{2n+2}(t+[\alpha])
\end{align*}
by making use of the same manner mentioned in Proposition \ref{prop1}.
Furthermore, by expanding these tau functions in terms of $\alpha$ and $\beta$ and compare with the coefficients, one can show \begin{align}
p_k(\tilde{\p}_t)p_l(-\tilde{\p}_t)\tau_{2n}(t,s)|_{s=-t}&=\htau_{2n}(t)p_k(\tilde{\p}_t)p_l(-\tilde{\p}_t)\htau_{2n}(t)\label{even},\\
p_k(\tilde{\p}_t)p_l(-\tilde{\p}_t)\tau_{2n+1}(t,s)|_{s=-t}&=p_l(-\tilde{\p}_t)\htau_{2n}(t)p_{k-1}(\tilde{\p}_t)\htau_{2n+2}(t)-p_{l-1}(-\tilde{\p}_t)\htau_{2n}(t)p_k(\tilde{\p}_t)\htau_{2n+2}(t).\nonumber
\end{align}
Therefore, these relations result in the Pfaff lattice hierarchy \cite{adler2002} (or DKP hierarchy \cite{jimbo1983})
\begin{align}\label{pfafflattice}
\begin{aligned}
0&=\oint_{C_\infty} \htau_{2n}(t-[z^{-1}])\htau_{2m+2}(t'+[z^{-1}])e^{\xi(t-t',z)}z^{2n-2m-2}dz\\
&+\oint_{C_0}\htau_{2n+2}(t+[z])\htau_{2m}(t'-[z])e^{\xi(t'-t,z^{-1})}z^{2n-2m}dz.
\end{aligned}
\end{align}
Interestingly, this bilinear identity could be written in terms of the wave functions
\begin{align}\label{pfaffwave}
\oint_{C_\infty}\Phi_{1}(t;z)\otimes h^{-1}\Phi_{2}(t';z^{-1})\frac{dz}{2\pi iz}+\oint_{C_0}\Phi_{2}(t;z)\otimes h^{-1}\Phi_{1}(t';z^{-1})\frac{dz}{2\pi iz}=0,
\end{align}
where wave functions $\Phi_1$ and $\Phi_2$ are defined by
\begin{align*}
\Phi_1(t,z)=e^{\sum_{i=1}^\infty t_i\Lambda^i}S(t)\chi(z),\quad \Phi_2(t,z)=h(t)e^{-\sum_{i=1}^\infty t_i\Lambda^{\top i}}S^{-\top}(t)\chi(z).
\end{align*}
It should be mentioned that with the explicit expressions of wave functions \cite[Theorem 3.2]{adler02}
\begin{subequations}
\begin{align}
\Phi_{1,2n}(t,z)&=e^{\sum_{i=1}^\infty t_iz^i}z^{2n}\frac{\tau_{2n}(t-[z^{-1}])}{\tau_{2n}(t)},\label{1e}\\
\Phi_{1,2n+1}(t,z)&=e^{\sum_{i=1}^\infty t_iz^i}z^{2n}\frac{(z+\p_{t_1})\tau_{2n}(t-[z^{-1}])}{\tau_{2n}(t)},\label{1o}\\
\Phi_{2,2n}(t,z)&=e^{-\sum_{i=1}^\infty t_iz^{-i}}z^{2n+1}h_{n}\frac{\tau_{2n+2}(t+[z])}{\tau_{2n+2}(t)},\label{2e}\\
\Phi_{2,2n+1}(t,z)&=-e^{-\sum_{i=1}^\infty t_iz^{-i}}z^{2n+1}h_n\frac{(z^{-1}-\p_{t_1})\tau_{2n+2}(t+[z])}{\tau_{2n+2}(t)},\label{2o}
\end{align}
\end{subequations}
the Pfaff lattice \eqref{pfafflattice} could be obtained by taking $\Phi_{1,2n}$ and $\Phi_{2,2m}$ into \eqref{pfaffwave}. Moreover, if one takes $t_1$-derivative or $t'_1$-derivative on both sides of $\Phi_{1,2n}$ and $\Phi_{2,2m}$, and employs equation \eqref{rel1}, the equation regarding $\Phi_{1,2n+1}$ and $\Phi_{2,2m+1}$ would be obtained.

Although these tau functions could be regarded as the most general case regarding the skew symmetric reduction, there should be another family of the tau functions about the odd-number index, which has been pointed out in \cite{vandeleur01}, to find out the tau functions of the large BKP hierarchy (or $D_\infty$ hierarchy in Sato school's framework \cite[Section 7]{jimbo1983}). In fact, there have been many discussions about the odd-indexed Pfaffian tau functions \cite{hu17,orlov16,vandeleur15,wang19} in recent years, but few are made about the wave functions. In the next parts, we would like to show the odd indexed Pfaffian tau functions would provide an alternative expressions for the odd indexed wave functions, and therefore to show the wave equations \eqref{pfaffwave} regarding \eqref{1o} and \eqref{2o} would lead  to the other cases of large BKP hierarchy.
\subsection{Partial Skew Orthogonal Polynomials, Skew Orthogonal Polynomials and Large BKP Hierarchy}\label{4.2}
Unlike the free charged fermions whose Fock space is generated by a single vacuum state, the Fock space of neutral fermions are generated by two different vacuum state---$|\text{vac}\rangle$ and $\phi_0|\text{vac}\rangle$ \cite{jimbo1983,wang19}. Therefore, the Fermion-Boson correspondence tells us the tau functions of the BKP and DKP hierarchy should be labelled by even and odd numbers respectively. Therefore, in \cite{hu17,orlov16,vandeleur15,wang19}, the authors have demonstrated the odd-indexed Pfaffian tau functions $\{\htau_{2n+1}\}_{n\in\mathbb{N}}$ for the BKP hierarchy or large BKP hierarchy from different perspectives.
The intimate connection between tau functions and orthogonal polynomials motivates the concept of the skew orthogonal polynomials related to the odd indexed tau functions. Some results are given in \cite{chang182}. 

Let's consider a skew symmetric bilinear inner product $\langle\,\cdot,\cdot\,\rangle$ from $\mathbb{R}[x]\times\mathbb{R}[y]\mapsto\mathbb{R}$, then the partial skew orthogonal polynomials $\{\tp_n(x)\}_{n\in\mathbb{N}}$ are defined under the this  inner product by
\begin{align}\label{psops}
\begin{aligned}
\tilde{P}_{2n}(x)&=\frac{1}{\htau_{2n}}\Pf\left(
\begin{array}{cc}
m_{i,j}&x^i\\
-x^j&0
\end{array}
\right)_{i,j=0}^{2n},\quad \htau_{2n}=\Pf(m_{i,j})_{i,j=0}^{2n-1}\\
 \tilde{P}_{2n+1}(x)&=\frac{1}{\htau_{2n+1}}\Pf\left(
 \begin{array}{ccc}
 0&\alpha_j&0\\
 -\alpha_i&m_{i,j}&x^i\\
 0&-x^j&0
 \end{array}
 \right)_{i,j=0}^{2n+1},\quad \htau_{2n+1}=\Pf\left(
 \begin{array}{cc}
 0&\alpha_i\\
 -\alpha_j&m_{i,j}
 \end{array}\right)_{i,j=0}^{2n},
\end{aligned}
\end{align}
where the general skew symmetric bi-moments 
\begin{align*}
\text{$m_{i,j}=\langle\, x^i,y^j\,\rangle:=\int_{\Gamma^2}x^iy^j\omega(x,y)d\mu(x)d\mu(y)$ with $\omega(x,y)=-\omega(y,x)$,}
\end{align*} and $\alpha_i$ is the single moments defined by $\alpha_i=\int_{\Gamma}x^id\mu(x)$. To make the moments are finite, we assume that measure $d\mu(x)$ decay fast at the boundary of the supports $\p\Gamma$. 

Remarkably, some specified choices of $\omega(x,y)$ have been indicated in \cite{orlov16}, well motivated from combinatorics, quantum Hall effect as well as different random matrix ensembles in cases of Mehta-Pandey interpolating ensemble, matrix models with orthogonal and symplectic symmetries, Bures ensemble, and so on. Later on, for simplicity, we denote $\Pf(i,j)=m_{i,j}$, $\Pf(d_0,i)=\alpha_i$, $\Pf(i,x)=x^i$ and $\Pf(d_0,x)=0$, and thus the notations in \eqref{psops} could be alternatively written as
\begin{align*}
\tp_{2n}(x)=\frac{1}{\htau_{2n}}\Pf(0,\cdots,2n,x),\quad \tp_{2n+1}(x)=\frac{1}{\htau_{2n+1}}\Pf(d_0,0,\cdots,2n+1,x)
\end{align*} 
with $\htau_{2n}=\Pf(0,\cdots,2n-1)$ and $\htau_{2n+1}=\Pf(d_0,0,\cdots,2n)$.

Firstly, we'd like to give some properties about the partial skew orthogonal polynomials and their connections with skew orthogonal polynomials. 
\begin{proposition}
Under the skew symmetric inner product $\langle\,\cdot,\cdot\,\rangle$, one has
\begin{align*}
&\langle \tp_{2n}(x),\tp_{2m}(y)\rangle=\langle \tp_{2n+1}(x),\tp_{2m+1}(y)\rangle=0,\\
&\langle \tp_{2n}(x),\tp_{2m+1}(y)\rangle=\frac{\htau_{2n+2}}{\htau_{2n}}\delta_{n.m},\quad \langle \tp_{2n+1}(x),\tp_{2m}(y)\rangle=-\frac{\htau_{2n+2}\htau_{2m+1}}{\htau_{2n+1}\htau_{2m}},\quad m\leq n.
\end{align*}
\end{proposition}
\begin{proof}
The proof of this proposition is based on the basic expansion of Pfaffians. Noting that
\begin{align*}
\langle \tp_{2n}(x),\tp_{2m}(y)\rangle&=\frac{1}{\htau_{2n}\htau_{2m}}\sum_{i=0}^{2n}\sum_{j=0}^{2m}(-1)^{i+j}\langle x^i,y^j\rangle\Pf(0,\cdots,\hat{i},\cdots,2n)\Pf(0,\cdots,\hat{j},\cdots,2m)\\
&=\frac{1}{\htau_{2n}\htau_{2m}}\sum_{i=0}^{2n}(-1)^i\Pf(0,\cdots,\hat{i},\cdots,2n)\Pf(0,\cdots,2m,i),
\end{align*}
where $\hat{i}$ means the missing of the index.
The last equation is equal to zero if $n\leq m$. Moreover, according to the skew symmetry, we know $\langle \tp_{2n}(x),\tp_{2m}(y)\rangle=0$ for all $n,\,m\in\mathbb{N}$. By the use of the same manner, one can prove the remaining formulae and we omit it here.
\end{proof}
Therefore, there is a strictly lower triangular matrix $\ts$, such that
\begin{align}\label{tildeh}
m_{\infty}=\ts^{-1}\tilde{h}\ts^{-\top},\quad \tilde{h}=\left(
\begin{array}{ccccc}
0&\frac{\htau_2}{\htau_0}&0&\frac{\htau_1\htau_4}{\htau_0\htau_3}&\cdots\\
-\frac{\htau_2}{\htau_0}&0&0&0&\cdots\\
0&0&0&\frac{\htau_4}{\htau_2}&\cdots\\
-\frac{\htau_1\htau_4}{\htau_0\htau_3}&0&-\frac{\htau_4}{\htau_2}&0&\cdots\\
\vdots&\vdots&\vdots&\vdots&\ddots
\end{array}\right)
\end{align}
and the vector of partial skew orthogonal polynomials can be written as $\tp(x)=\ts\chi(x)$. From the skew Borel decomposition $m_\infty=S^{-1}hS^{-\top}$, we know $\tilde{h}$ and $h$ are congruent via the equation $(S\ts^{-1})\tilde{h}(S\ts^{-1})^\top=h$. Moreover, there exists a relation between skew orthogonal polynomials $\{P_n(x)\}_{n\in\mathbb{N}}$ given by \eqref{sops} and partial skew orthogonal polynomials $\{\tp_n(x)\}_{n\in\mathbb{N}}$
\begin{align*}
P(x)=S\chi(x)=(S\ts^{-1})\ts\chi(x)=(S\ts^{-1})\tp(x).
\end{align*}
In what follows, we would like to show the form of $S\ts^{-1}$. 
\begin{proposition}\label{ss}
For the skew orthogonal polynomials $\{P_n(x)\}_{n\in\mathbb{N}}$ and partial skew orthogonal polynomials $\{\tp_n(x)\}_{n\in\mathbb{N}}$, there exist
\begin{align*}
P_{2n}(x)&=\tp_{2n}(x),\\
P_{2n+1}(x)&=\tp_{2n+1}(x)+\alpha_{n}\tp_{2n}(x)+\beta_n\tp_{2n-1}(x)
\end{align*}
with proper coefficients $\alpha_n$ and $\beta_n$.
\end{proposition}
\begin{proof}
It is obvious that the polynomials of even order are exactly the same in these two cases according to the definitions of the Pfaffian elements. For the polynomials of odd order, we need to use the Pfaffian identity
\begin{align*}
&\Pf(d_0,\star,2n,2n+1,x)\Pf(\star)\\
&=\Pf(d_0,\star,2n)\Pf(\star,2n+1,x)-\Pf(d_0,\star,2n+1)\Pf(\star,2n,x)+\Pf(d_0,\star,x)\Pf(\star,2n,2n+1)
\end{align*}
with $\{\star\}=\{0,\cdots,2n-1\}$. By taking the explicit expressions of these Pfaffians in \eqref{sops} and \eqref{psops},  we make this proof completed with
\begin{align}\label{abn}
\alpha_n=\frac{\Pf(d_0,0,\cdots,2n-1,2n+1)}{\htau_{2n+1}},\quad \beta_n=-\frac{\htau_{2n+2}\htau_{2n-1}}{\htau_{2n}\htau_{2n+1}}.
\end{align}
\end{proof}
Following this proposition, one can show 
\begin{align*}
S\ts^{-1}=\left(\begin{array}{ccccc}
1&&&&\\
\alpha_0&1&&&\\
0&0&1&&\\
0&\beta_1&\alpha_1&1&\\
\vdots&\vdots&\vdots&\vdots&\ddots
\end{array}
\right)
\end{align*}
and denote this matrix as $\mathcal{A}$. In addition, if we introduce the Lax operator $L=S\Lambda S^{-1}$ and operator $\tilde{L}=\ts\Lambda\ts^{-1}$, then we have the following relation between these two operators
\begin{align}\label{laxrelation}
L=\mathcal{A}\tilde{L}\mathcal{A}^{-1}.
\end{align}

Next, we would introduce the time parameters and assume the commuting vector fields $\p_{t_n}m_\infty=\Lambda^n m_\infty+m_\infty\Lambda^{\top n}$, which connects the partial skew orthogonal polynomials (wave functions) with tau function.
\begin{proposition}\label{sopwave}
With the time evolutions $\p_{t_n}m_\infty=\Lambda^n m_\infty+m_\infty\Lambda^{\top n}$ and $\p_{t_n}\alpha=\Lambda^n\alpha$, we have
\begin{align}\label{psop2}
\tp_n(x;t)=\frac{\htau_n(t-[x^{-1}])}{\htau_n(t)}x^n.
\end{align}
\end{proposition}
\begin{proof}
The expansion of the notation in \eqref{psop2} gives us
\begin{align*}
\tp_{n}(x;t)=\frac{e^{\xi(\tilde{\p}_t,x^{-1})}\htau_n(t)}{\htau_n(t)}x^n=\sum_{k\geq0}\frac{p_k(-\tilde{\p}_t)\htau_n}{\htau_n}x^{n-k},
\end{align*}
where $\{p_k\}_{k\geq0}$ are the Schur functions. Moreover, by the use of the time evolution condition, one knows \cite[Equation 3.84]{hirota04}
\begin{align*}
\p_{t_j}\Pf(0,\cdots,2n-1)&=\sum_{k=0}^{2n-1}\Pf(0,\cdots,k+j,\cdots,2n-1),\\
\p_{t_j}\Pf(d_0,0,\cdots,2n)&=\sum_{k=0}^{2n}\Pf(d_0,0,\cdots,k+j,\cdots,2n).
\end{align*}
Therefore, by expanding the Schur function, one can demonstrate
\begin{align*}
p_k(-\tilde{\p}_t)\htau_{2n}&=(-1)^k\Pf(0,\cdots,\widehat{2n-k},\cdots,2n),\\
p_k(-\tilde{\p}_t)\htau_{2n+1}&=(-1)^k\Pf(d_0,0,\cdots,\widehat{2n+1-k},2n+1)
\end{align*}
and $p_k(-\tilde{\p}_t)\htau_n=0$ if $k>n$. This formula corresponds to the Pfaffian form of the partial skew orthogonal polynomials in \eqref{psops}.
\end{proof}
Therefore, we know that $\alpha_n$ in \eqref{abn} could be easily expressed as $\alpha_n=\p_{t_1}\log\htau_{2n+1}$ if time parameters are involved. Moreover, the equation \eqref{rel1} and Proposition \ref{ss} imply
\begin{align*}
(x+\p_{t_1})\tp_{2n}(x;t)=\tp_{2n+1}(x;t)+\p_{t_1}\log\frac{\htau_{2n+1}}{\htau_{2n}}\tp_{2n}(x;t)-\frac{\htau_{2n+2}\htau_{2n-1}}{\htau_{2n+1}\htau_{2n}}\tp_{2n-1}(x;t)
\end{align*} 
with respect to the $t_1$-flow. In fact, for the partial skew orthogonal polynomials $\{\tp_n(x;t)\}_{n\in\mathbb{N}}$ with odd and even indexes, they obey the same recurrence relation \cite[Lemma 3.21]{chang182}
\begin{align}\label{substitution}
(x+\p_{t_1})\tp_n(x;t)=\tp_{n+1}(x)+\left(\p_{t_1}\log \frac{\htau_{n+1}}{\htau_n}\right)\tp_n(x;t)-\frac{\htau_{n+2}\htau_{n-1}}{\htau_{n+1}\htau_{n}}\tp_{n-1}(x;t).
\end{align}
Substituting the expression \eqref{psop2} into equation \eqref{substitution}, one can find
\begin{align}\label{largebkp}
\begin{aligned}
&x^2\left[\htau_{n+1}(t)\htau_n(t-[x^{-1}])-\htau_n(t)\htau_{n+1}(t-[x^{-1}])\right]\\
&+x\left[\htau_{n+1}(t)\p_{t_1}\htau_n(t-[x^{-1}])-\p_{t_1}\htau_{n+1}(t)\htau_n(t-[x^{-1}])\right]+\htau_{n+2}(t)\htau_{n-1}(t-[x^{-1}])=0.
\end{aligned}
\end{align}
This equation is a particular Fay identity of the large BKP hierarchy given by \cite[Corollary 6.1]{vandeleur01}.
Again, by expanding the Miwa variables into Schur functions, i.e. $\htau_n(t-[x^{-1}])=\sum_{k\geq0}p_k(-\tilde{\p}_t)\htau_n z^{-k}$, we finally arrive at
\begin{align*}
\begin{aligned}
\htau_n\cdot p_{k+2}(-\tilde{\p}_t)\htau_{n+1}&-\htau_{n+1}\cdot p_{k+2}(-\tilde{\p}_t)\htau_n+\p_{t_1}\htau_{n+1}\cdot p_{k+1}(-\tilde{\p}_t)\htau_n\\
&-\htau_{n+1}\cdot\p_{t_1}p_{k+1}(-\tilde{\p}_t)\htau_n-\htau_{n+2}\cdot p_k(-\tilde{\p}_t)\htau_{n-1}=0,\quad \forall k\in\mathbb{Z}.
\end{aligned}
\end{align*}
The first nontrivial example is the $k=0$ case
\begin{align*}
(D_2+D_1^2)\htau_{n}\cdot\htau_{n+1}=2\htau_{n-1}\htau_{n+2},
\end{align*}
which is the first nontrivial example of large BKP hierarchy proposed in \cite{kac97} and an example in \cite[Equation 7.7]{jimbo1983}. 
Moreover, this particular Fay identity gives us an alternative expressions of the odd-indexed tau functions.
 \begin{proposition}\label{waveal}
The wave functions $\Phi_{1,2n+1}(t;z)$ and $\Phi_{2,2n+1}(t;z)$ could be expressed in terms of odd-indexed tau functions via
\begin{subequations}
\begin{align}
&\Phi_{1,2n+1}(t;z)=e^{\sum_{i=1}^\infty t_iz^i}\label{phi1o}\\
&\times\left(\frac{1}{\htau_{2n+1}}\htau_{2n+1}(t-[z^{-1}])z^{2n+1}+\frac{\p_{t_1}\htau_{2n+1}}{\htau_{2n}\htau_{2n+1}}\htau_{2n}(t-[z^{-1}])z^{2n}-\frac{\htau_{2n+2}}{\htau_{2n}\htau_{2n+1}}\htau_{2n-1}(t-[z^{-1}])z^{2n-1}\right),\nonumber\\
&\Phi_{2,2n+1}(t;z)=-e^{-\sum_{i=1}^\infty t_iz^{-i}}\label{phi2o}\\
&\times\left(\frac{\htau_{2n+2}}{\htau_{2n}\htau_{2n+1}}\htau_{2n+1}(t+[z])z^{2n}-\frac{\p_{t_1}\htau_{2n+1}}{\htau_{2n}\htau_{2n+1}}\htau_{2n+2}(t+[z])z^{2n+1}-\frac{1}{\htau_{2n+1}}\htau_{2n+3}(t+[z])z^{2n+2}\right).\nonumber
\end{align}
\end{subequations}
\end{proposition}

\begin{proof}
The first equation \eqref{phi1o} is obvious if one substitutes the Fay identity \eqref{largebkp} into the expression \eqref{1o}. Moreover, if one takes the shift $t-[x^{-1}]\mapsto t$ and transform $x^{-1}\mapsto x$, one can get an equivalent Fay identity
\begin{align*}
&x^{-2}\left[\htau_{n+1}(t+[x])\htau_n(t)-\htau_n(t+[x])\htau_{n+1}(t)\right]\\
&+x^{-1}\left[\htau_{n+1}(t+[x])\p_{t_1}\htau_n(t)-\p_{t_1}\htau_{n+1}(t+[x])\htau_n(t)\right]+\htau_{n+2}(t+[x])\htau_{n-1}(t)=0,
\end{align*}
which implies  $$(x^{-1}-\p_{t_1})\htau_{n}(t+[x])=x^{-1}\htau_{n-1}(t+[x])\frac{\htau_{n}}{\htau_{n-1}}-\p_{t_1}\log\htau_{n-1}\htau_{n}(t+[x])-x\htau_{n+1}(t+[x])\frac{\htau_{n-2}}{\htau_{n-1}}$$
up to a shift in ``$n$'' index.
Taking this equation into \eqref{2o}, one can get the conclusion.
\end{proof}
Therefore, if one takes the odd-indexed Pfaffian wave functions into \eqref{pfaffwave} of the form
\begin{align*}
\oint_{C_\infty}\Phi_{1,2n+1}(t;z)\Phi_{2,2m}(t';z^{-1})\frac{dz}{2\pi iz}+\oint_{C_0}\Phi_{2,2n+1}(t;z)\Phi_{1,2m}(t';z^{-1})\frac{dz}{2\pi iz}=0,
\end{align*}
one can obtain
\begin{align*}
\htau_{2n}(t)&\left(\oint_{C_\infty}e^{\xi(t-t',z)}\htau_{2n+1}(t-[z^{-1}])\htau_{2m+2}(t'+[z^{-1}])z^{2n-2m-1}dz\right.\\
&\left.+\oint_{C_0}e^{\xi(t'-t,z^{-1})}\htau_{2n+3}(t+[z])\htau_{2m}(t'-[z])z^{2n-2m+1}dz\right)\\
-\htau_{2n+2}(t)&\left(\oint_{C_\infty}
e^{\xi(t-t',z)}\htau_{2n-1}(t-[z^{-1}])\htau_{2m+2}(t'+[z^{-1}])z^{2n-2m-3}dz\right.\\
&\left.+\oint_{C_0}e^{\xi(t'-t,z^{-1})}\htau_{2n+1}(t+[z])\htau_{2m}(t'-[z])z^{2n-2m-1}dz
\right)=0,
\end{align*}
and this equation is valid for all $t,\,t'\in\mathbb{C}$. With $n$ and $m$ being symmetry invariant in this case, this equation results in
\begin{align}
\begin{aligned}\label{oddeven}
\htau_{2n}(t)\htau_{2m+1}(t')&=\oint_{C_\infty}e^{\xi(t-t',z)}\htau_{2n-1}(t-[z^{-1}])\htau_{2m+2}(t'+[z^{-1}])z^{2n-2m-3}\\
&+\oint_{C_0}e^{\xi(t'-t,z^{-1})}\htau_{2n+1}(t+[z])\htau_{2m}(t'-[z])z^{2n-2m-1}dz.
\end{aligned}
\end{align}
Furthermore, from the equation \eqref{oddeven} and Proposition \ref{waveal}, we know the equation
\begin{align*}
\oint_{C_\infty}\Phi_{1,2n+1}(t;z)\Phi_{2,2m+1}(t';z^{-1})\frac{dz}{2\pi iz}+\oint_{C_0}\Phi_{2,2n+1}(t;z)\Phi_{1,2m+1}(t';z^{-1})\frac{dz}{2\pi iz}=0
\end{align*}
could be expressed in terms of tau functions as
\begin{align*}
0&=\oint_{C_\infty}e^{\xi(t-t',z)}\htau_{2n-1}(t-[z^{-1}])\htau_{2m+1}(t'+[z^{-1}])z^{2n-2m-2}dz\\&+\oint_{C_0}e^{\xi(t'-t,z^{-1})}
\htau_{2n+1}(t+[z])\htau_{2m-1}(t'-[z])z^{2n-2m}dz.
\end{align*}
To conclude, we have the following proposition.
\begin{proposition}
The large BKP hierarchy
\begin{align}\label{lbkp}
\begin{aligned}
\frac{1}{2}(1-(-1)^{n+m})\htau_n(t)\htau_{m}(t')&=\oint_{C_\infty}e^{\xi(t-t',z)}\htau_{n-1}(t-[z^{-1}])\htau_{m+1}(t'+[z^{-1}])z^{n-m-2}dz\\&+\oint_{C_0}e^{\xi(t'-t,z^{-1})}\htau_{n+1}(t+[z])\htau_{m-1}(t'-[z])z^{n-m}dz
\end{aligned}
\end{align}
could be written in terms of wave functions \eqref{pfaffwave}.
\end{proposition}

\subsection{Rank Two Shift Condition and B-Toda Lattice}\label{4.3}
From the last two subsections, we know the Pfaff lattice hierarchy and the large BKP hierarchy are the same if the odd-indexed tau functions are introduced. However, it should be mentioned that one of the most important features of BKP hierarchy hasn't been involved. It was shown by Hirota \cite[\S 3.3]{hirota04} that the tau functions of BKP hierarchy always admit the form
\begin{align*}
\tau=\Pf(1,2,\cdots,2n),\quad \Pf(i,j)=\int_{-\infty}^{t_1} D_{t_1} f_i(t)\cdot f_j(t)dt_1,
\end{align*}
and $\{f_i(t),\,i=1,\cdots,2n\}$ should satisfy the linear differential relation $\p_{t_n}f_i(t)=\p_{t_1}^n f_i(t)$. If we introduce the label $d_i$ by $\Pf(d_n,i)=\p_{t_1}^n f_i(x)$, then the derivative formula 
\begin{align}\label{bkp}
\p_{t_1}\Pf(i,j)=\Pf(d_0,d_1,i,j),\quad \Pf(d_0,d_1)=0
\end{align}
results in the Gram-type Pfaffian elements.
It is not difficult to see that Pfaffian tau functions shown in \eqref{psops} are not the tau functions of BKP hierarchy if we merely consider the time evolutions $\p_{t_n}m_\infty=\Lambda^nm_\infty+m_\infty\Lambda^{\top n}$ and $\p_{t_n}\alpha=\Lambda^n\alpha$. 
In fact, with the commuting vector fields, the derivative formula \eqref{bkp} implies the relations between the single moments and bi-moments, which can be formulated by the following rank two shift condition
\begin{align}\label{ranktwo}
\Lambda m_\infty+m_\infty \Lambda^\top=\Lambda \alpha\alpha^\top-\alpha\alpha^\top\Lambda^\top, \quad \alpha=(\alpha_0,\alpha_1,\cdots)^\top,\quad \alpha_i(t)\in C^\infty.
\end{align}
\begin{remark}
The rank two shift condition can be naturally realised by the moments of Bures ensemble
\begin{align*}
m_{i,j}=\int_{\mathbb{R}_+^2}\frac{x-y}{x+y}x^iy^jd\mu(x)d\mu(y)
\end{align*}
from the relation $$m_{i+1,j}+m_{i,j+1}=\int_{\mathbb{R}_+}x^{i+1}d\mu(x)\int_{\mathbb{R}_+}y^jd\mu(y)-\int_{\mathbb{R}_+}x^id\mu(x)\int_{\mathbb{R}_+}y^{j+1}d\mu(y).$$
\end{remark}
The first aim is to derive the B-Toda lattice from the 2d-Toda theory, by the consideration of rank two shift condition.
By using the skew-Borel decomposition $m_\infty=S^{-1}hS^{-\top}$ in \eqref{sop} and denoting $L=S\Lambda S^{-1}$, $\sigma=S\alpha$ and $\rho=S\Lambda\alpha$, the above rank two shift condition is equal to
\begin{align*}
Lh+hL^\top=\rho\sigma^\top-\sigma\rho^\top,
\end{align*}
where the right hand side is again a skew symmetric matrix of rank two.
Moreover, we have the following proposition.
\begin{proposition}\label{prop:4.1}
The rank two shift condition is equal to
\begin{subequations}
\begin{align}
\rho_{2n+1}\sigma_{2n-1}-\rho_{2n-1}\sigma_{2n+1}&=h_{n-1}l_{2n+1,2n-2}-h_n,\label{btoda1}\\
\rho_{2n-1}\sigma_{2n-2}-\rho_{2n-2}\sigma_{2n-1}&=-h_{n-1}(l_{2n-1,2n-1}+l_{2n-2,2n-2}),\label{btoda2}\\
\rho_{j-1}\sigma_{2n-1}-\rho_{2n-1}\sigma_{j-1}&=h_{n-1}l_{j-1,2n-2},\quad\,\,\,\,\, j=2n+1,2n+3,\cdots,\label{btoda3}\\
\rho_{j-1}\sigma_{2n-2}-\rho_{2n-2}\sigma_{j-1}&=-h_{n-1}l_{j-1,2n-1},\quad j=2n+1,2n+2,\cdots,\label{btoda4}
\end{align}
\end{subequations}
where $\{l_{i,j}\}$ are the elements of the Lax operator $L$.
\end{proposition}
\begin{proof}
One could easily observe $h=[-h_0e_2,h_0e_1,-h_1e_4,h_1e_3,\cdots]$, where $e_i$ is the $i$th elementary column vector, which means its $i$th element is $1$ and the others are $0$. Therefore, the $(i,2j-1)$-th and $(i,2j)$-th elements of $Lh$ could be computed as
\begin{align*}
Lh(i,2j-1)
&=\left\{\begin{array}{ll}
-h_{j-1}l_{i-1,2j-1},&i\geq 2j,\\
-h_{j-1},&i=2j-1,\end{array}
\right.\quad
Lh(i,2j)
&=\left\{\begin{array}{ll}
h_{j-1}l_{i-1,2j-2},&i\geq 2j-1,\\
h_{j-1},&i=2j-2.\end{array}
\right.
\end{align*}
Similarly, if one denotes $h=[h_0e_2^\top,-h_0e_1^\top,h_1e_4^\top, -h_1e_3^\top,\cdots]$, then one could get
\begin{align*}
hL^\top(2i-1,j)
&=\left\{\begin{array}{ll}
h_{i-1}l_{j-1,2i-1},&j\geq 2i,\\
h_{i-1},&j=2i-1,\end{array}
\right.\quad
hL^\top(2i,j)
&=\left\{\begin{array}{ll}
-h_{i-1}l_{j-1,2i-2},&j\geq 2i-1,\\
-h_{i-1},&j=2i-2.\end{array}
\right.
\end{align*}
Making a combination of these results, one can easily arrive at the conclusion.
\end{proof}
Furthermore, from the 2d-Toda theory (see equations \eqref{lij} and \eqref{stilde}), we know the explicit expressions for $l_{j,j}$ could be written as 
\begin{align*}
l_{j,j}=s_{j,j-1}+\tilde{s}_{j+1,j}=s_{j,j-1}-s_{j+1,j},\quad
s_{j,j-1}=\frac{p_1(-\tilde{\p}_t)\tau_j(t,s)|_{s=-t}}{\tau_j(t,s)|_{s=-t}},
\end{align*}
and by the use of the relation \eqref{even} with $k=0,\,l=0$ and $k=0,\,l=1$, the equation \eqref{btoda2} could be rewritten as
\begin{align}\label{bt1}
D_{t_1}\htau_{2n}\cdot\htau_{2n-2}=\htau_{2n-2}^2(\rho_{2n-1}\sigma_{2n-2}-\rho_{2n-2}\sigma_{2n-1}):=\hr_{2n-1}\hs_{2n-2}-\hr_{2n-2}\hs_{2n-1},
\end{align}
where $\hr_{2n-1}=\htau_{2n-2}\rho_{2n-1}$, $\hr_{2n-2}=\htau_{2n-2}\rho_{2n-2}$, $\hs_{2n-1}=\htau_{2n-2}\sigma_{2n-1}$ and $\hs_{2n-2}=\htau_{2n-2}\sigma_{2n-2}$.

On the other hand, if one takes $j=2n+1$ in \eqref{btoda4} and then this equation equals to
\begin{align*}
\rho_{2n}\sigma_{2n-2}-\rho_{2n-2}\sigma_{2n}=-h_{n-1}l_{2n,2n-1}.
\end{align*}
In 2d-Toda theory, it is shown that (c.f. \eqref{l3})
\begin{align*}
l_{2n,2n-1}=\frac{\tau_{2n}(t,s)|_{s=-t}\p_{t_1}^2\tau_{2n}(t,s)|_{s=-t}-(\p_{t_1}^2\tau_{2n}(t,s)|_{s=-t})^2}{(\tau_{2n}(t,s)|_{s=-t})^2}.
\end{align*}
Now, taking \eqref{even} into consideration and choosing $(k,\,l)$ as $(0,\,0)$, $(0,\,1)$ and $(1,\,1)$ respectively, one can find
\begin{align*}
l_{2n,2n-1}=\frac{\htau_{2n}\p_{t_1}^2\htau_{2n}-(\p_{t_1}\htau_{2n})^2}{\htau_{2n}^2}.
\end{align*}
Therefore, equation \eqref{btoda4} in this case could be rewritten as
\begin{align}\label{bt2}
\frac{1}{2}D_{t_1}^2\htau_{2n}\cdot\htau_{2n}=\hr_{2n-2}\hs_{2n}-\hr_{2n}\hs_{2n-2}.
\end{align}
One should notice that $\hat{\sigma}$ and $\hat{\rho}$ are interrelated.
\begin{proposition}
One has
\begin{align}\label{sr2}
\p_{t_1}(\htau\sigma)=\htau\rho,
\end{align}
where $\htau=\text{diag}\{\htau_0 I_{2\times2}, \htau_2 I_{2\times 2},\cdots\}$.
\end{proposition}

\begin{proof}
If one substitutes $\sigma=S\alpha$ and $\rho=S\Lambda\alpha$ into equation \eqref{sr2}, this equation equals to show that
\begin{align}\label{sla}
\p_{t_1}(\htau S\alpha)=\htau S\Lambda\alpha.
\end{align}
By noting that $\alpha$ obeys the derivation rule $\p_{t_n}\alpha=\Lambda^n\alpha$, the equation \eqref{sla} is equivalent to prove
$
[\p_{t_1}(\htau S)]\alpha=0.
$
Observing that $S$ is the dressing operator of the skew orthogonal polynomials, we can rewrite this equation as
\begin{align*}
\int_{\mathbb{R}_+}\p_{t_1}(\htau P(x;t))d\mu(x;t)=0,
\end{align*}
where $P(x;t)$ is the vector of skew orthogonal polynomials related to the skew-Borel decomposition $m_\infty=S^{-1}hS^{-\top}$.
Since the moment matrix satisfies the rank two shift condition, one can get 
\begin{align}\label{zero}
\begin{aligned}
\p_{t_1}(\htau_{2n}P_{2n}(x;t))&=\Pf(d_0,d_1,0,\cdots,2n,x),\\
\p_{t_1}(\htau_{2n}P_{2n+1}(x;t))&=\Pf(d_0,d_1,0,\cdots,2n-1,2n+1,x),
\end{aligned}
\end{align}
and the integration over $x$ contributes to another label $d_0$. From the definition of Pfaffian, one knows these results are equal to zero, which completes the proof.
\end{proof}

\begin{corollary}
$\alpha$ and $\Lambda\alpha$ are the null vectors of $\p_{t_1}(\htau S)$. Equivalently, $\sigma$ and $\rho$ are the null vectors of $\p_{t_1}(\htau S)S^{-1}$.
\end{corollary}
\begin{proof}
In the last proposition, we have proved that $\alpha$ is the null vector of $\p_{t_1}(\htau S)$.
Furthermore, if we consider the integration of $x\p_{t_1}(\htau P(x;t))$ over $x$,  the label $d_1$ would appear in expressions \eqref{zero} by taking the place of label $x$, which is also equal to zero. This property demonstrates that
\begin{align*}
\int_{\mathbb{R}_+}x\p_{t_1}(\htau P(x;t))d\mu(x;t)=0,
\end{align*}
and one can get $\p_{t_1}(\htau S)\Lambda\alpha=0$. 
\end{proof}

Therefore, the B-Toda lattice \eqref{bt1} and \eqref{bt2} could be written as
\begin{subequations}
\begin{align}
D_t\htau_{2n}\cdot\htau_{2n-2}&=D_t\hat{\sigma}_{2n-1}\cdot\hat{\sigma}_{2n-2}\label{bbt1}\\
D_t^2\htau_{2n}\cdot\htau_{2n}&=2D_t\hat{\sigma}_{2n-2}\cdot\hat{\sigma}_{2n}.\label{bbt2}
\end{align}
\end{subequations}
Also, the odd-indexed $\hat{\sigma}_{2n+1}$ is related to $\hat{\sigma}_{2n}$.
Since $P(x)=S\chi(x)$, one could write $\sigma$ as
\begin{align*}
\sigma=S\alpha=\int_{\mathbb{R}_+}P(x;t)d\mu(x;t),
\end{align*}
and
\begin{align*}
\hat{\sigma}_{2n}=\int_{\mathbb{R}_+}\htau_{2n}P_{2n}(x;t)d\mu(x;t),\quad \hat{\sigma}_{2n+1}=\int_{\mathbb{R}_+}\htau_{2n}P_{2n+1}(x;t)d\mu(x;t).
\end{align*}
From the definition, one can find
\begin{align*}
\hat{\sigma}_{2n}=\Pf(d_0,0,\cdots,2n):=\htau_{2n+1},\quad
\hat{\sigma}_{2n+1}={\Pf(d_0,0,\cdots,2n-1,2n+1)}
\end{align*}
if we substitute the explicit expressions of skew orthogonal polynomials $\{P_{n}(x)\}_{n\in\mathbb{N}}$ into \eqref{sops}. From the relationship in Proposition \ref{sopwave}, we know $\hat{\sigma}_{2n+1}=\p_{t_1}\htau_{2n+1}$. Therefore, the first member of the B-Toda hierarchy could be written as
\begin{subequations}
\begin{align}
&D_{t_1}^2\htau_{2n}\cdot\htau_{2n}=2D_{t_1}\htau_{2n-1}\cdot\htau_{2n+1},\label{b1}\\
&2D_{t_1}\htau_{2n}\cdot\htau_{2n+2}=D_{t_1}^2\htau_{2n+1}\cdot\htau_{2n+1}.\label{b2}
\end{align}
\end{subequations}
if one takes $\{\hat{\sigma}_n\}_{n\in\mathbb{N}}$ into \eqref{bbt1} and \eqref{bbt2}.
Interestingly, these two equations could be expressed in a unified way as $D_{t_1}^2\htau_n\cdot\htau_n=2D_{t_1}\htau_{n-1}\cdot\htau_{n+1}$, which is called as the B-Toda lattice in the known literature \cite{hirota01}. This fact demonstrates that $\{\htau_{2n+1}\}_{n\geq0}$ (or $\{\sigma_n,\,\rho_n\}_{\geq0}$) in this case are no longer auxiliary variables but another family of tau functions of BKP hierarchy \cite{hu17,loris99}, generated by $\phi_0|\text{vac}\rangle$ which we mentioned at the beginning of this section.

As mentioned, the rank two shift condition is also motivated by the moments of Bures ensemble, that is,  if we consider a skew symmetric bilinear inner product $\langle\cdot,\cdot\rangle:\mathbb{R}[x]\times\mathbb{R}[y]\to \mathbb{R}$ such that
\begin{align}\label{bmoment}
\langle x^i,y^j\rangle:=m_{i,j}=\int_{\mathbb{R}_+^2}\frac{x-y}{x+y}x^iy^j\omega(x)\omega(y)dxdy.
\end{align}
In this case, the partition function of Bures ensemble is defined as the Pfaffian of the moment matrix and therefore
\begin{align*}
\htau_{2n}=\Pf(m_{i,j})_{i,j=0}^{2n-1}=(-1)^n\int_{\xi^{2n}}\prod_{1\leq i<j\leq 2n}\frac{(x_j-x_i)^2}{x_j+x_i}\prod_{i=1}^{2n}\omega(x_i)dx_i
\end{align*}
by making applications of the de Bruijn formula (See \cite[Equation A.5]{chang182}) and Schur's Pfaffian identity (See \cite[Equation A.8]{chang182}) and $\xi^{2n}$ is the subspace of $\mathbb{R}^{2n}_+$ defined by $\xi^{2n}=\{(x_1,\cdots,x_{2n})\,|\, 0<x_1<\cdots<x_{2n}\}$. It is also noted that the odd-indexed tau functions could be also written as the partition functions of Bures ensemble for the odd-number particles \cite{chang182,forrester16}
\begin{align*}
\htau_{2n+1}&=\Pf\left(\begin{array}{cc}
0&\alpha_j\\
-\alpha_i&m_{i,j}
\end{array}\right)_{i,j=0}^{2n}\\
&=(-1)^n\int_{\xi^{2n+1}}\prod_{1\leq i<j\leq 2n+1}\frac{(x_j-x_i)^2}{x_j+x_i}\prod_{i=1}^{2n+1}\omega(x_i)dx_i,
\end{align*}
where $\xi^{2n+1}$ is the configuration space of $\mathbb{R}_+^{2n+1}$ with the ordered set $\xi^{2n+1}=\{(x_1,\cdots,x_{2n+1})\,|\,0<x_1<\cdots<x_{2n+1}\}$.

\subsection{B-Toda Hierarchy}
In the last subsection, we embed the B-Toda hierarchy into 2d-Toda theory and obtain the equations from the rank two shift condition. It is shown that these equations \eqref{b1} and \eqref{b2} are not included in the large BKP hierarchy, thus motivating us to consider what is the sub-hierarchy of the B-Toda lattice.
To this end, firstly we restate the derivatives of partial skew orthogonal polynomials under the rank two shift condition \cite[Corollary 3.16]{chang182}.

\begin{lemma}\label{tderivative}
Under the rank two shift condition, there is a derivative formula for the partial skew orthogonal polynomials
\begin{align}\label{derpsop}
\begin{aligned}
\htau_{n+1}\cdot&\p_{t_1}(\htau_{n-1}\tp_{n-1}(x;t))\\
&=\p_{t_1}\htau_{n+1}\cdot\htau_{n-1}\tp_{n-1}(x;t)-\htau_{n}\cdot\p_{t_1}(\htau_{n}\tp_n(x;t))+\p_{t_1}\htau_n\cdot\htau_n\tp_n(x;t).
\end{aligned}
\end{align}
\end{lemma}
In particular, if we take \eqref{psop2} into the above equation, it leads to a novel particular Fay identity
\begin{align}\label{btodafay1}
\begin{aligned}
\htau_{n+1}\cdot\p_{t_1}\htau_{n-1}(t-[x^{-1}])&+x\htau_n\cdot\p_{t_1}\htau_n(t-[x^{-1}])\\
&=\p_{t_1}\htau_{n+1}\cdot\htau_{n-1}(t-[x^{-1}])+x\p_{t_1}\htau_n\cdot\htau_{n}(t-[x^{-1}]).
\end{aligned}
\end{align}
By the use of $\htau_n(t-[x^{-1}])=\sum_{k\geq0}p_k(-\tilde{\p}_t)\htau_n x^{-k}$, then
\begin{align*}
\htau_{n+1}\cdot\p_{t_1}p_k(-\tilde{\p}_t)\htau_{n-1}-\p_{t_1}\htau_{n+1}\cdot p_k(-\tilde{\p}_t)\tau_{n-1}+\htau_n\cdot\p_{t_1}p_{k+1}(-\tilde{\p}_t)\htau_n-\p_{t_1}\htau_n\cdot p_{k+1}(-\tilde{\p}_t)\htau_n=0
\end{align*}
is valid for all $k\in\mathbb{N}$. The first nontrivial example is given by $k=0$ and it can be explicitly expressed as 
\begin{align*}
D_t^2\htau_n\cdot\htau_n=2D_t\htau_{n-1}\cdot\htau_{n+1},
\end{align*}
which is the B-Toda lattice in \eqref{b1} and \eqref{b2}.
Moreover, we'd like to obtain the general bilinear expressions for this hierarchy. It requires detailed analysis of the wave functions.

From the particular Fay identity \eqref{btodafay1}, one can find
\begin{align}\label{relation1}
\begin{aligned}
\frac{\p_{t_1}\htau_{n}(t-[x^{-1}])}{\htau_{n+1}}x^n&=\p_{t_1}\log\htau_n\frac{\htau_n(t-[x^{-1}])}{\htau_{n+1}}x^n\\
&+\sum_{i=1}^n(-1)^{n-i}\p_{t_1}\log\frac{\htau_{i+1}}{\htau_{i-1}}\frac{\htau_{i-1}(t-[x^{-1}])}{\htau_i}x^{i-1}.
\end{aligned}
\end{align}
With the shift $t-[x^{-1}]\mapsto t$ and transform $x^{-1}\mapsto x$, the Fay identity \eqref{btodafay1} is changed into
\begin{align*}
x\htau_{n+1}(t+[x])\p_{t_1}\htau_{n-1}+\htau_n(t+[x])\p_{t_1}\htau_n=x\p_{t_1}\htau_{n+1}(t+[x])\htau_{n-1}+\htau_n\p_{t_1}\htau_n(t+[x]),
\end{align*}
which implies
\begin{align}\label{relation2}
\frac{\p_{t_1}\htau_{n+1}(t+[x])}{\htau_n}x^n&=\p_{t_1}\log\htau_{n-1}\frac{\htau_{n+1}(t+[x])}{\htau_n}x^n+\sum_{i=2}^n(-1)^{n-i}\left(
\p_{t_1}\log\frac{\htau_{i}}{\htau_{i-2}}
\right)\frac{\htau_i(t+[x])}{\htau_{i-1}}x^{i-1}\nonumber\\
&+(-1)^{n-1}\left(
\p_{t_1}\log\htau_1\cdot\htau_1(t+[x])-\p_{t_1}\htau_1(t+[x])
\right),
\end{align}
where the tail term comes from the assumption $\htau_{-1}=0$ and $\htau_0=1$. Since we have obtained the derivatives of $\htau_n(t\pm[x])$, we can state the following proposition.

\begin{proposition}\label{btodah}
Under the rank two shift condition, the bilinear identity for the wave functions \eqref{pfaffwave} could be written as
\begin{align}\label{btodahierarchy}
\begin{aligned}
(-1)^{n+m}&\left(\p_{t_1}\htau_n(t)\cdot\htau_m(t')-\htau_n(t)\cdot\p_{t'_1}\htau_m(t')\right)\\&=\oint_{C_0}e^{\xi(t'-t,z^{-1})}z^{n-m-1}\htau_{n+1}(t+[z])\htau_{m-1}(t'-[z])dz\\&-\oint_{C_\infty}e^{\xi(t-t',z)}z^{n-m-1}\htau_{n-1}(t-[z^{-1}])\htau_{m+1}(t'+[z^{-1}])dz.
\end{aligned}
\end{align}
\end{proposition}
The proof of this proposition needs the following lemma.
\begin{lemma}\label{lem1}
Under the assumption that $\htau_{-1}=0$, we have
\begin{align}
\oint_{C_0}e^{\xitt}z^{-m}\p_{t_1}\htau_1(t+[z])\htau_{m-1}(t'-[z])dz=(-1)^{m-1}\p_{t'}\htau_m(t').
\end{align}
\end{lemma}
\begin{proof}
Taking $n=0$ in the large BKP hierarchy \eqref{lbkp} and taking the $t_1$-derivative on the both sides, one knows
\begin{align*}
\oint_{C_0}e^{\xitt}z^{-m}\p_{t_1}\htau_1(t+[z])\htau_{m-1}(t'-[z])dz=\oint_{C_0}e^{\xitt}z^{-m-1}\htau_1(t+[z])\htau_{m-1}(t'-[z])dz
\end{align*}
due to $\htau_{-1}=0$ and $\htau_0=1$.
Noting that the right hand side in the above equation can be expressed as
\begin{align*}
\frac{1}{2}(1-(-1)^m)\p_{t'_1}\htau_m(t')-\oint_{C_0}e^{\xitt}z^{-m}\htau_1(t+[z])\p_{t'_1}\htau_{m-1}(t'-[z])dz
\end{align*}
if we take the $t'_1$-derivative of \eqref{lbkp} when $n=0$. By the use of equation \eqref{relation1}, one could find the latter term is equal to 
\begin{align*}
\frac{1}{2}(1-(-1)^m)\htau_m\p_{t'_1}\log\htau_{m-1}+(-1)^{m-1}\htau_m\sum_{i=1}^{m-1}(-1)^i(1-(-1)^i)\p_{t'_1}\log\frac{\htau_{i+1}}{\htau_{i-1}}=\frac{1}{2}(1+(-1)^m)\p_{t'_1}\htau_m,
\end{align*}
and thus complete the proof.
\end{proof}
Now we could state the proof of Proposition \ref{btodah}.
\begin{proof}
If one takes the $t_1$-derivative on the both sides of large BKP hierarchy, then one can find
\begin{align*}
\frac{1}{2}&(1-(-1)^{n+m})\p_{t_1}\htau_n(t)\cdot\htau_m(t')\\
&=\oint_{C_\infty}e^{\xit}\htau_{n-1}(t-[z^{-1}])\htau_{m+1}(t'+[z^{-1}])z^{n-m-1}dz\\
&-\oint_{C_0}e^{\xitt}\htau_{n+1}(t+[z])\htau_{m-1}(t'-[z])z^{n-m-1}dz\\
&+\oint_{C_\infty}e^{\xit}\p_{t_1}\htau_{n-1}(t-[z^{-1}])\htau_{m+1}(t'+[z^{-1}])z^{n-m-2}dz\\
&+\oint_{C_0}e^{\xitt}\p_{t_1}\htau_{n+1}(t+[z])\htau_{m-1}(t'-[z])z^{n-m}dz.
\end{align*}
Taking the expressions of $\p_{t_1}\htau_{n-1}(t-[z^{-1}])$ and $\p_{t_1}\htau_{n+1}(t+[z])$ from \eqref{relation1} and \eqref{relation2}, and by the use of the large BKP hierarchy \eqref{lbkp} again, we know the last two terms are equal to
\begin{align*}
\frac{1}{2}(1+(-1)^{n+m})\htau_m(t')\p_{t_1}\htau_n(t)+(-1)^n\htau_n\oint_{C_0}e^{\xitt}\p_{t_1}\htau_1(t+[z])\htau_{m-1}(t'-[z])z^{-m}dz.
\end{align*}
By using the result of lemma \ref{lem1}, the proof is completed.
\end{proof}

It should be mentioned that the B-Toda hierarchy's Fay identity \eqref{relation1} could give us the following relation
\begin{align}\label{tpsop}
\p_{t_1}\tp_{n}(x;t)=\frac{\htau_{n+1}}{\htau_n}\sum_{i=1}^n(-1)^{n-i}\left(\p_{t_1}\log\frac{\htau_{i+1}}{\htau_{i-1}}\right)\frac{\htau_{i-1}}{\htau_i}\tp_{i-1}(x;t).
\end{align}
Making a combination of the results \eqref{substitution} and \eqref{tpsop}, one could find the spectral problem
\begin{align}\label{spectral}
\begin{aligned}
x\tp_n(x;t)&=\tp_{n+1}(x;t)+\left(
\p_{t_1}\log\frac{\htau_{n+1}}{\htau_n}
\right)\tp_n(x;t)+(\p_{t_1}^2\log\htau_n-\frac{\htau_{n+2}\htau_{n-1}}{\htau_{n+1}\htau_n})\tp_{n-1}(x;t)\\
&-\frac{\htau_{n+1}}{\htau_n}\sum_{i=1}^{n-1}(-1)^{n-i}\left(\p_{t_1}\log\frac{\htau_{i+1}}{\htau_{i-1}}\right)\frac{\htau_{i-1}}{\htau_i}\tp_{i-1}(x;t)
\end{aligned}
\end{align}
with regard to the partial skew orthogonal polynomials.
The equation \eqref{spectral} also implies the following four term recurrence relation
\begin{align}\label{ftrr}
\begin{aligned}
&x(\tilde{P}_n(x)+u_n\tp_{n-1}(x))\\
&=\tp_{n+1}(x)+(v_{n+1}-v_n+u_n)\tp_n(x)-u_n(v_{n+1}-v_n+u_{n+1})\tp_{n-1}(x)-u_n^2u_{n-1}\tp_{n-2}(x)
\end{aligned}
\end{align}
for all $n\in\mathbb{N}$,
where $u_n$ and $v_n$ have the expressions
\begin{align*}
u_n=\frac{\htau_{n+1}\htau_{n-1}}{\htau_n^2},\quad v_n=\p_{t_1}\log\htau_n
\end{align*}
if the $t_1$-parameter is introduced. 
Therefore, one can denote the spectral operator $\tilde{L}$ as
\begin{align*}
\tilde{L}=(\Lambda^0+\Lambda^{-1}a)^{-1}(\Lambda^{1}+\Lambda^0b+\Lambda^{-1}c+\Lambda^{-2}d),
\end{align*}
such that $x\tp=\tilde{L}\tp$ and $a=\di(a_1,a_2,\cdots)$, $b=\di(b_0,b_1,\cdots)$, $c=\di(c_1,c_2,\cdots)$, $d=\di(d_2,d_3,\cdots)$ with elements $a_n=u_n$, $b_n=v_{n+1}-v_n+u_n$, $c_n=-u_n(v_{n+1}-v_n+u_{n+1})$ and $d_n=-u_n^2u_{n-1}$. 

Till now, we show that the tau function of the B-Toda hierarchy is obtained by a skew symmetric reduction from the 2d-Toda hierarchy as well as a rank two shift condition. The first condition tells us the B-Toda lattice could be viewed as a special case of the Pfaff lattice, and therefore, B-Toda lattice inherits the merits of the Pfaff lattice. Moreover, the rank two shift condition suggests us another basis of polynomials, which satisfy a four term recurrence relation and provide us an explicit spectral problem. Therefore, we would utilise these two important features and give the Lax representation with explicit Lax matrix for the B-Toda hierarchy.

\begin{proposition}
Given a skew symmetric moment matrix $m_\infty$ with the rank two shift condition \eqref{ranktwo} and its related skew-Borel decomposition $m_\infty=S^{-1}hS^{-\top}$, one can define a family of skew orthogonal polynomials $P(x)=S\chi(x)$ satisfying the spectral problem $LP(x)=xP(x)$ with $L=S\Lambda S^{-1}$, and $L$ could be expressed as 
\begin{align*}
L=(\Lambda^0+\Lambda^{-1}\alpha+\Lambda^{-2}\beta)(\Lambda^0+\Lambda^{-1}a)^{-1}(\Lambda^{1}+\Lambda^0b+\Lambda^{-1}c+\Lambda^{-2}d)(\Lambda^0+\Lambda^{-1}\alpha+\Lambda^{-2}\beta)^{-1}
\end{align*}
with $\alpha=\di(v_1,0,v_3,0,v_5,\cdots)$ and $\beta=\di(0,-u_1u_2,0,-u_3u_4,\cdots)$. Moreover, one can get the Lax integrability of the B-Toda hierarchy from the skew-Borel decomposition
\begin{align*}
\p_{t_n} L=[-(L^n+hL^{\top n}h^{-1})_{<0},L]:=[-\pi_i(L^n),L],
\end{align*}
where $\pi_i$ is the projection to the strict lower triangular matrix part in $2\times 2$ block matrix sense.
\end{proposition}
In general, the Lax matrix $L$ for the skew-orthogonal polynomials could hardly be written in an explicit form except the first few diagonals, see \cite[Equation 0.17]{adler02} and \cite[Equation 1.5]{kodama10}. However, this proposition tells us, under the rank two shift condition, the Lax matrix $L$ could be written explicitly and therefore one can get more information about the Lax integrability of the B-Toda hierarchy.


\section{Concluding remarks}
This article is about two new sub-hierarchies by considering the symmetric and skew symmetric reductions on the general moment matrix and some rank shift conditions, motivated by the Cauchy two-matrix model and Bures ensemble. Although the first members of these hierarchies have been achieved by the average characteristic polynomials approach, the moment matrix method provides us an insight to the integrable structures as well as algebraic structures for the whole discrete hierarchy.  Moreover, the rank shift condition has never appeared before in the studies of integrable systems, and we hope this kind of concepts would bring some new features to the studies, not only in the classical integrable system theory but in random matrix theory.

However, some problems are left for future studies. One is regarding the partial skew orthogonal relation. Notice that the matrix $\tilde{h}$ in \eqref{tildeh} could be written as
\begin{align*}
\tilde{h}=\htau(\Lambda-\Lambda^\top)^{-1}\htau,\quad \htau=\text{diag}\left(
\frac{\htau_1}{\htau_0},\frac{\htau_2}{\htau_1},\cdots
\right),
\end{align*}
which is simply structured. It is unclear what is the role they will play in the random matrix theory and it is also difficult to see how do these odd-indexed tau functions come from the matrix decomposition theory and how do they connect with 2d-Toda's tau functions.
In addition, similar with the relationship between Toda and KP hierarchy, these two hierarchies could be viewed as the discrete versions of BKP and CKP hierarchies. The first members of these discrete flows, in fact, have been partially demonstrated in \cite[Equation (27)]{loris99} and \cite[Equation (31)]{loris992} in the studies of symmetry reductions from the continuous hierarchies and  Darboux transformations. Being lack of the discrete variables, equations in the above-mentioned references are not closed in the $t_1$-flow. However, They did provide some hints to reveal the relations between these discrete and continuous hierarchies. This kind of link may help us to consider the Virasoro constraints for these discrete hierarchies, and further to the gap probabilities for the Bures ensemble and Cauchy two-matrix model, regarding the continuous hierarchies have already been known, e.g. the Adler-Shiota-van Moerbeke formula and Virasoro symmetries for the BKP hierarchy have been shown in \cite{vandeleur95} and non-linearizable Virasoro symmetries for the CKP hierarchy has been given in \cite{chang13}. It is also interesting to consider the general spectral problem implied in \cite{zuo19} with Lax matrix
\begin{align*}
L=(\Lambda-a_{l+2})^{-1}(\Lambda^{k+1}+a_1\Lambda^k+\cdots+a_{l+1}\Lambda^{k-l}),\, 1\leq k<l.
\end{align*} In fact, in the recent work \cite{forrester192}, Forrester and one of the author proposed a $\theta$-deformed Cauchy two-matrix model with the joint probability density function
\begin{align*}
\frac{\prod_{1\leq i<j\leq N}(x_j-x_i)(y_j-y_i)}{\prod_{i,j=1}^N(x_i+y_j)}\prod_{1\leq i<j\leq N}(x_j^\theta-x_i^\theta)(y_j^\theta-y_i^\theta)\prod_{i=1}^N \omega_1(x_i)\omega_2(y_i)dx_idy_i,
\end{align*}
which probably relates to the above spectral problem since the generalised $\theta$-deformed Cauchy bi-orthogonal polynomials will provide a longer recurrence relation and we will show it in a following work.

\section*{Acknowledgement}
S. Li  is supported by the ARC Centre of Excellence for Mathematical and Statistical frontiers (ACEMS) and G. Yu is supported by National Natural Science Foundation of China (Grant no. 11871336).

\small
\bibliographystyle{abbrv}

\begin{thebibliography}{10}

\bibitem{adler1999}
M.~Adler, E.~Horozov, and P.~van Moerbeke.
\newblock {The Pfaff lattice and skew-orthogonal polynomials}.
\newblock {\em Int. Math. Res. Notices}, 11: 569--588, 1999.

\bibitem{adler97}
M. Adler and P. van Moerbeke.
\newblock{Group factorization, moment matrices, and Toda lattices}.
\newblock {\em Int. Math. Res. Notices}, 12: 555-572, 1997.


\bibitem{adler19992}
M.~Adler and P. van Moerbeke.
\newblock{The spectrum of coupled random matrices}.
\newblock {\em Ann. of Math (2)}, 149: 921-976, 1999.

\bibitem{adler00}
M. Adler and P. van Moerbeke.
\newblock{Hermitian, symmetric and symplectic random ensembles: PDEs for the distribution of the spectrum}.
\newblock{\em Ann. of Math (2)}, 153: 149-189, 2001.

\bibitem{adler02}
M. Adler and P. van Moerbeke.
\newblock Toda versus Pfaff lattice and related polynomials.
\newblock {\em Duke Math J.}, 112: 1-58, 2002.



\bibitem{adler03}
M. Adler and P. van Moerbeke.
\newblock{Recursion relations for unitary integrals, combinatorics and the Toeplitz lattice}.
\newblock {\em Comm. Math. Phys.}, 237: 397-440, 2003.

\bibitem{adler2002}
M.~Adler, T.~Shiota, and P.~van Moerbeke.
\newblock Pfaff $\tau$-functions.
\newblock {\em Math. Ann.}, 322(3): 423--476, 2002.



\bibitem{adler18}
M. Adler and P. van Moerbeke.
\newblock The AKS theorem, A.C.I. systems and random matrix theory.
\newblock {\em J. Phys. A}, 51: 423001, 2018



\bibitem{bertola2009}
M.~Bertola, M.~Gekhtman, and J.~Szmigielski.
\newblock The {C}auchy two-matrix model.
\newblock {\em Commun. Math. Phys.}, 287(3): 983--1014, 2009.

\bibitem{bertola2010} 
M. Bertola, M. Gekhtman and J. Szmigielski.
\newblock{Cauchy biorthogonal polynomials}.
\newblock {\em J. Approx. Theory}, 162: 832-867, 2010. 

\bibitem{bertola2014} 
M. Bertola, M. Gekhtman and J. Szmigielski.
\newblock{Cauchy-Laguerre two-matrix model and the Meijer-G random point field}.
\newblock {\em Commum. Math. Phys.}, 326: 111-144, 2014. 

\bibitem{brini17}
A. Brini, G. Carlet, S. Romano and P. Rossi.
\newblock{Rational reductions of the 2D-Toda hierarchy and mirror symmetry}.
\newblock {\em J. Eur. Math. Soc.}, 19: 835-880, 2017.

\bibitem{bures69}
D. Bures.
\newblock {An extension of Kakutani's thoerem on infinite product measures to the tensor product of semifinite w*-algebras}.
\newblock {\em Trans. of the AMS}, 135: 199-212, 1969.



\bibitem{chang13}
L. Chang and C. Wu.
\newblock Tau function of the CKP hierarchy and non-linearizable Virasoro symmetries.
\newblock {\em Nonlinearity}, 26: 2577, 2013.


\bibitem{chang182}
X. Chang, Y. He, X. Hu and S. Li.
\newblock{Partial-skew-orthogonal polynomials and related integrable lattices with Pfaffian tau-functions.}
\newblock {\em Comm. Math. Phys.}, 364: 1069-1119, 2018.

\bibitem{chang18}
X. Chang, X. Hu and S. Li.
\newblock{Degasperis-Procesi peakon dynamical system and finite Toda lattice of CKP type}.
\newblock {\em Nonlinearity}, 31: 4746, 2018.

\bibitem{deift00}
P. Deift.
\newblock{Orthogonal Polynomials and Random Matrices: A Riemann-Hilbert Approach}.
\newblock Courant Lecture Notes 3, American Mathematical Society, 2000.

\bibitem{faybusovich99}
L. Faybusovich and M. Gekhtman.
\newblock{On Schur flows.}
\newblock {\em J. Phys. A}, 32: 4671, 1999.

\bibitem{forrester16}
P. Forrester and M. Kieburg.
\newblock{Relating the Bures measure to the Cauchy two-matrix model}.
\newblock{\em Commun. Math. Phys.}, 342: 151-187, 2016.

\bibitem{forrester192}
P. Forrester and S. Li.
\newblock{Fox H-kernel and $\theta$-deformation of the Cauchy two-matrix model and Bures ensemble}.
\newblock {\em Int. Mat. Res. Not.}, rnz028, 2019.


\bibitem{forrester19}
P. Forrester and S. Li.
\newblock{Classical discrete symplectic ensembles on the linear and exponential lattice: skew orthogonal polynomials and correlation functions}.
\newblock{ arXiv: 1902.09042}, accepted by Trans. of AMS,  2019.

\bibitem{forrester06}
P. Forrester and N. Witte.
\newblock{Bi-orthogonal polynomials on the unit circle, regular semi-classical weights and integrable systems}.
\newblock {Constr. Approx.}, 24: 201-237, 2006.


\bibitem{gerasimov90}
A. Gerasimov, A. Marshakov, A. Mironov, A. Morozov and A. Orlov.
\newblock{Matrix models of two-dimensional gravity and Toda theory}.
\newblock{\em Nucl. Phys. B}, 357: 565-618, 1990.

\bibitem{hirota04}
R. Hirota (Translated by A. Nagai, J. Nimmo and C. Gilson).
\newblock The direct method in soliton theory.
\newblock Cambridge University Press, 2004.

\bibitem{hirota01}
R. Hirota, M. Iwao and S. Tsujimoto.
\newblock{Soliton equations exhibiting Pfaffian solutions}.
\newblock{\em Glasgow Math. J.}, 43A: 33-41, 2001.


\bibitem{hu06}
X. Hu, J. Zhao and C. Li.
\newblock {Matrix integrals and several integrable differential-difference systems}.
\newblock {\em J. Phys. Soc. Jpn.}, 75: 054003, 2006.


\bibitem{hu17}
X. Hu and S. Li.
\newblock{The partition function of Bures ensemble as the $\tau$-function of BKP and DKP hierarchies: continuous and discrete}.
\newblock {\em J. Phys. A}, 50: 285201, 2017.



\bibitem{jimbo1983}
M. Jimbo and T. Miwa.
\newblock{Solitons and infinite dimensional Lie algebras}.
\newblock{\em Publ. RIMS, Kyoto Univ.}, 19: 943-1001, 1983.

\bibitem{kac97}
V. Kac and J. van de Leur.
\newblock The geometry of spinors and the multicomponent BKP and DKP hierarchies.
\newblock{\em CRM Proc. Lecture notes 14}, AMS, Providence, 159-202, 1998.

\bibitem{kodama10}
Y. Kodama and V. Pierce.
\newblock{The Pfaff lattice on symplectic matrices}.
\newblock{\em J. Phys. A}, 43: 055206, 2010.


\bibitem{li2019}
C. Li and S. Li.
\newblock {The Cauchy Two-Matrix Model, C-Toda Lattice and CKP Hierarchy}.
\newblock{\em J. Nonlinear Sci.}, 29: 3-27, 2019.


\bibitem{loris992}
I. Loris.
\newblock On reduced CKP equations.
\newblock {\em Inverse Problems}, 15: 1099-1109, 1999.

\bibitem{loris99}
I. Loris and R. Willox.
\newblock{Symmetry reductions of the BKP hierarchy}.
\newblock{\em J. Math. Phys.}, 40: 1420-1431, 1999.


\bibitem{lundmark2005}
H.~Lundmark and J.~Szmigielski.
\newblock Degasperis-{P}rocesi peakons and the discrete cubic string.
\newblock {\em Int. Math. Res. Pap.}, 2005(2): 53--116, 2005.

\bibitem{manas19}
M. Ma$\tilde{n}$as.
\newblock Revisiting biorthogonal polynomials. An LU factorization discussion.
\newblock arXiv: 1907.04280v1, 2019.

\bibitem{mukaihira02}
A. Mukaihira and Y. Nakamura.
\newblock {Schur flow for orthogonal polynomials on the unit circle and its integrable discretization.}
\newblock {J. Comp. Appl. Math.}, 139: 75-94, 2002.


\bibitem{orlov16}
A. Orlov, T. Shiota and K. Takasaki.
\newblock Pfaffian structures and certain solutions to BKP hierarchies II. Multiple integrals.
\newblock arXiv: 1611.02244, 2016.


\bibitem{sommers03}
H. Sommers and K. Zyczkowski.
\newblock Bures volume of the set of mixed quantum states.
\newblock {\em J. Phys. A}, 36: 10083-10100, 2003.

\bibitem{takasaki18}
K. Takasaki.
\newblock Toda hierarchies and their applications.
\newblock {\em J. Phys. A}, 51: 203001, 2018.


\bibitem{ueno84}
K. Ueno and K. Takasaki.
\newblock Toda lattice hierarchy.
\newblock {\em Adv. Studies Pure Math.}, 4: 1-95, 1984.


\bibitem{vandeleur95}
J. van de Leur.
\newblock{The Adler-Shiota-van Moerbeke formula for the BKP hierarchy.}
\newblock {\em J. Math. Phys.}, 36: 4940, 1995.

\bibitem{vandeleur01}
J. van de Leur.
\newblock{Matrix integrals and the geometry of spinors}.
\newblock{\em J. Nonlinear Math. Phys.}, 8: 288-310, 2001.

\bibitem{vandeleur15}
J. van de Leur and A. Orlov.
\newblock{Pfaffian and determinantal tau functions}.
\newblock {\em Lett. Math. Phys.}, 105: 1499-1531, 2015.



\bibitem{wang19}
Z. Wang and S. Li.
\newblock{BKP hierarchy and Pfaffian point process}.
\newblock{\em Nucl. Phys. B}, 939: 447-464, 2019.


\bibitem{zuo19}
D. Zuo.
\newblock Frobenius manifolds and a new class of extended affine Weyl groups $\tilde{W}^{(k,k+1)}(A_l)$.
\newblock{ arXiv: 1905.09470, 2019.}

\end{thebibliography}

\def\cydot{\leavevmode\raise.4ex\hbox{.}}
  \def\cydot{\leavevmode\raise.4ex\hbox{.}} \def\cprime{$'$}

\end{document}